\documentclass[8pt]{article}

\def\Sch{L}
\def\Ric{\mbox{Ric}}
\def\Scal{\mbox{Scal}}
\def\Hess{\mbox{Hess}}
\def\grad{\mbox{grad}}
\def\gE{g_E}
\def\ovgE{\overline{g}_E}
\def\B{{\mathcal B}}
\def\covE{{\mathcal D}}
\newcommand{\scri}{\mathscr{I}}

\def\xizero{\widehat{\xi}}
\def\omegazero{\widehat{\omega}}
\def\psizero{\widehat{\Psi}}
\def\derpsizero{\widehat{\Psi}}

\def\ar{\mathrm{a}}
\def\br{\mathrm{b}}

\newcommand{\bm}[1]{\mbox{\boldmath $#1$}}

\newcommand{\ov}[1]{\overline{#1}}

\usepackage[utf8]{inputenc} 
\usepackage{amsmath}
\usepackage{mathrsfs}
\usepackage{amsfonts}
\usepackage{xcolor}
\usepackage{caption}
\usepackage{amsthm}
\usepackage{enumerate}
\usepackage{amssymb}
\usepackage{authblk}
\usepackage{mathtools}
\mathtoolsset{showonlyrefs=true} 

\usepackage{psfrag}
\usepackage[colorinlistoftodos]{todonotes}

\DeclareFontEncoding{LS1}{}{}
\DeclareFontSubstitution{LS1}{stix}{m}{n}
\DeclareSymbolFont{symbols2}{LS1}{stixfrak} {m} {n}
\DeclareMathSymbol{\operp}{\mathbin}{symbols2}{"A8}

\usepackage{xcolor}
 	\definecolor{carblue}{rgb}{0.30, 0.50, 1}

\newcommand{\car}[1]{\textcolor{black}{#1}}

\newcommand{\lr}[1]{\left( #1 \right)}
\newcommand{\lrbrace}[1]{\left\lbrace #1 \right\rbrace}

\newcommand{\sign}{\mathrm{sign}}

\def\la{\langle}
\def\ra{\rangle}

\newcommand{\conf}[1]{\mathrm{Conf}\lr{#1}}

\newtheorem{theorem}{Theorem}[section]
\newtheorem{proposition}{Proposition}[section]

\newtheorem{lemma}[theorem]{Lemma}
\newtheorem{remark}{Remark}[section]
\newtheorem{definition}{Definition}[section]

\title{Covariant classification of conformal Killing vectors of locally conformally flat $n$-manifolds with an application to Kerr-de Sitter}

\author[1]{Marc Mars\thanks{\it marc@usal.es}}
\author[2]{Carlos Peón-Nieto\thanks{\it carlos.peon-nieto@matfyz.cuni.cz}}
\affil[1]{\it Instituto de F\'{\i}sica  Fundamental y Matem\'aticas, Universidad de Salamanca \\
Plaza de la Merced s/n 37008, Salamanca, Spain}
\affil[2]{\it Institute of Theoretical Physics, Faculty of Mathematics and Physics,
Charles University, V Hole\v{s}ovi\v{c}k\'ach 2, 180 00 Prague 8, Czech Republic}
\date{\today}

\begin{document}

\maketitle 

\begin{abstract}
We obtain a coordinate independent algorithm to determine the class of conformal Killing vectors of a locally conformally flat $n$-metric $\gamma$ of signature $(r,s)$ modulo conformal transformations of $\gamma$. This is done in terms of endomorphisms in the pseudo-orthogonal Lie algebra $\mathfrak{o}(r+1,s+1)$ up to conjugation of \car{the} group $O(r+1,s+1)$. The explicit classification is worked out in full for the Riemannian $\gamma$ case ($r = 0, s = n$). As an application of this result, we prove that the set of five dimensional, $(\Lambda>0)$-vacuum, algebraically special metrics with non-degenerate optical matrix, analyzed  in \cite{reall15} is in one-to-one correspondence with the metrics in the Kerr-de Sitter-like class. This class \cite{Kdslike,marspeonKSKdS21} exists in all dimensions and its defining properties involve only properties at $\scri$.  The equivalence between two seemingly unrelated classes of metrics  
points towards interesting connections between the algebraically special type of the bulk spacetime and the conformal geometry at null infinity.
\end{abstract}

\section{Introduction}

Conformal invariance plays a fundamental role in many physical theories that include critical phenomena, conformal field theories or electromagnetism, among many others. Conformal Killing vectors (CKV) in \car{(a locally conformally)} flat space, being the infinitesimal generators of (local) conformal transformations,
are therefore also of great importance.  The conformal group induces a natural equivalence relation between CKVs in flat space. Two such CKVs are said to be equivalent if there is a (local) conformal transformation that maps one to another.
Conformal invariance in a theory means that the relevant object is
the conformal class of CKVs instead of  individual CKV representatives in the class. 

\car{As summarized in more detail below \cite{IntroCFTschBook} (also \cite{Kdslike,marspeon20,marspeon21}), the classification of conformal classes of CKVs in a flat $n$-dimensional space (of signature $(r,s)$) is equivalent to the algebraic classification of
skew-symmetric endomorphisms in an $(n+2)$-dimensional flat space of signature $(r+1,s+1)$ (i.e.  elements of the Lie algebra $\mathfrak{o}(r+1,s+1)$)
up to pseudo-orthogonal $O(r+1,s+1)$ transformations. The latter classification is worked out in detail in \cite{marspeon21} by means of only elementary linear algebra methods, although it is worth to remark that, in a more algebraic language, this is a particular case of classification of semisimple adjoint orbits, which is a well-known problem in Lie theory for which there actually exists a general framework (e.g. \cite{adjointorbs,knapp}). }

\car{For any classification result of the adjoint orbits of $\mathfrak{o}(r+1,s+1)$ to become of practical use to our case at hand, one first needs to find a map  between the algebras of CKVs and $\mathfrak{o}(r+1,s+1)$. Such map is easily constructed in Cartesian coordinates  \cite{IntroCFTschBook} (also \cite{Kdslike,marspeon20,marspeon21}), but a general coordinate independent approach appears to be missing in the literature. This is specially relevant in physical contexts where general covariance is a key ingredient (e.g.
general relativity), as it is often the case that the quantities are expressed in coordinate systems that are convenient for the problem at hand, and hence a priori unrelated to any Cartesian description. 
The main objective of this paper is to
provide a simple, algorithmic and coordinate independent 
classification scheme for conformal classes of CKVs in locally conformally flat manifolds of arbitrary signature. Our main result is given in
Theorem \ref{theoendocords}.}

As already said, this result can be of interest in any physical problem
(in a locally conformally flat space) where conformal invariance and diffeormorphism invariance play a crucial role.  An example of paramount importance is
the study of the asymptotic properties of spacetimes. A precise definition of the asymptotic region of a spacetime $(\widetilde M, \widetilde g)$
can be given in terms of conformal scalings, as long as the metric satisfies the property of being {\it conformally extendable}. We say that a metric $\widetilde g$ is conformally extendable if there exists a metric $g = \Omega^2 \widetilde g$ for a (sufficiently) smooth positive function $\Omega$ on $\widetilde M$, such that $g$ admits a (sufficiently) smooth extension
to $M := \widetilde M \cup \partial \widetilde M$, where $\partial \widetilde M = \{ \Omega = 0 \}$ is usually called {\it null infinity} and denoted $\scri$. A milestone result in general relativity by H. Friedrich \cite{friedrich81,friedrich81bis} proves, using the conformal properties of the spacetime, that there exists a well-posed Cauchy problem with data on $\scri$ for the $(\Lambda>0)$-vacuum case in four dimensions \cite{Fried86geodcomp}. Similar results also hold in higher dimensions \cite{Anderson2005,andersonchrusciel05,Kaminski21,kichenassamy03}, based on a different formalism. More details on these asymptotic Cauchy problems are given in Section \ref{secKdSlike}. 
Just like in the ordinary initial value problem of the Einstein equations \cite{BeigChruscielKID97}, the presence of Killing vectors (KV) in the spacetime constraint the initial data of the asymptotic Cauchy problem with  $\Lambda>0$ \cite{KIDPaetz,marspeondata21}. For a KV $X$ of $\widetilde g$, it is a general fact that $X$ extends to $\scri$ as a tangent vector $\xi$, which is necessarily a conformal Killing vector of the induced metric $\gamma$ at $\scri$. In four spacetime dimensions, a Killing initial data (KID) equation was obtained in \cite{KIDPaetz} for the asymptotic Cauchy problem of the $(\Lambda>0)$-vacuum equations. This gives a necessary and sufficient condition that $\xi$ and the initial data must satisfy in order for the spacetime metric to admit a KV $X$, which coincides with $\xi$ at $\scri$. The KID equation was extended to higher dimensions in \cite{marspeondata21} where it was also  proved that, when the data are restricted to be analytic, it gives necessary and sufficient conditions for $\xi$ to extend to a spacetime KV.

The fact that the asymptotic Cauchy problem for the $(\Lambda>0)$-vacuum equations is formulated in terms of conformal metrics, leaves a conformal gauge freedom in the initial data (cf. Section \ref{secKdSlike}).  Moreover, as a covariant theory, the data are also equivalent under diffeomorphisms of $\scri$. This has the interesting consequence that the set of  conformal diffeomorphisms (or conformal isometries) of $\scri$ is a symmetry at $\scri$ in the sense above. More specifically, the conformal group  $\conf{\scri}$ is the
set of diffeomorphisms $\varphi: \scri \longrightarrow \scri$ satisfying $\varphi^\star(\gamma) = \omega ^2 \gamma$ for some smooth positive function $\omega$ of $\scri$. We show in Section \ref{secKdSlike} that if $\xi$ satisfies the KID equation for certain data, then the vector fields of the form $\varphi_\star(\xi)$ for each $\varphi \in \conf{\scri}$ also satisfy the KID equation
and correspond to the same symmetry of $\widetilde g$.  Thus, the symmetries of $\widetilde g$ are in correspondence with the conformal classes of CKVs $[\xi]$ instead of with specific representatives in the class. As a consequence, whenever the geometry
at $\scri$ is locally conformally flat, the algorithmic classification result achieved in Theorem \ref{theoendocords} is of direct applicability.

In the second part of the paper we apply this theorem to establish the equivalence between two a priori unrelated families of
$(\Lambda>0)$-vacuum solutions of the Einstein field equations in five dimensions. The first class \car{consists of} the
algebraically special spacetimes with non-degenerate optical matrix, classified in \cite{reall15}. The second is the class
of so-called Kerr-de Sitter-like class of metrics with conformally flat $\scri$. This class is a natural generalization of Kerr-de Sitter, first defined in four dimensions in \cite{KdSnullinfty,Kdslike} and later generalized to arbitrary dimensions in
\cite{marspeonKSKdS21}
in the locally conformally flat $\scri$ case. The characterizing property of this class is that
its asymptotic data is constructed canonically from a CKV at $\scri$. In fact, there is exactly one Kerr-de Sitter-like metric associated to each conformal class of CKVs of the metric at $\scri$. Therefore, the space of conformal classes $[\xi]$ provides a good representation of the moduli space of metrics in the class. The explicit form of the metrics is the Kerr-de Sitter class in all dimensions was obtained in  \cite{marspeonKSKdS21} via a rather unexpected equivalence (in all dimensions) between this class and the family of general $(\Lambda>0)$-vacuum spacetimes of Kerr-Schild type and satisfying a natural fall-off condition at infinity. The fact that this family (in five dimensions) is also equivalent to the family obtained in
\cite{reall15} indicates that there might be interesting and unexpected connections for $(\Lambda>0)$-vacuum spacetimes (in arbitrary dimension higher than four) between being algebraically special and being conformally extendable and having very special properties
at null infinity. We emphasize that the family studied in \cite{reall15} made {\it no} a priori assumption on the asymptotic properties of the spacetime.

Moreover, it is worth to emphasize that the metrics in \cite{reall15} were heuristically found to be either Kerr-de Sitter or a limit thereof. 
This same fact holds for the Kerr-de Sitter-like class \cite{marspeonKSKdS21}, but it can be seen as a natural consequence of the topological structure of the space of conformal classes of CKVs $\xi$, as well as the well-posedness of the asymptotic Cauchy problem. This perspective strengthens the uniqueness result of Kerr-de Sitter as understood in \cite{reall15}, in the sense that it proves that no further limits can be obtained from Kerr-de Sitter without, at least, substantially modifying the asymptotic properties.

The plan of the paper is as follows. In Section \ref{seccoords} we prove our main result, Theorem \ref{theoendocords}, which provides a method for a coordinate independent classification of conformal classes of CKVs of any locally conformally flat $n$-manifold $(\Sigma,\gamma)$ of signature $(r,s)$. We do this in terms of the (simpler) algebraic classification of skew-symmetric endomorphisms of an $(n+2)$-dimensional flat manifold $(\mathbb{R}^{r+1,s+1},\eta)$ up to isometries of $\eta$, where $\eta$ is a flat metric of signature $(r+1,s+1)$. The latter amounts to the classification of the Lie algebra $\mathfrak{o}(r+1,s+1)$ up to adjoint action of the Lie group $O(r+1,s+1)$, i.e. the equivalence classes $[F] = \{ F' \in \mathfrak{o}(r+1,s+1) \mid F' = \Lambda \cdot F \cdot \Lambda^{-1}, ~ \forall \Lambda \in  O(r+1,s+1) \} $, where the dot stands for the usual matrix multiplication. A key result for Theorem \ref{theoendocords} is (cf. Proposition \ref{propxiF}) that we find a way to assign an element $F \in \mathfrak{o}(r+1,s+1)$ to any CKV $\xi$ of $\gamma$ based solely on pointwise properties of $\xi$ (and its derivatives). In addition, for later use we give in subsection \ref{secriem} the explicit definition of a set of quantities which uniquely determines a class $[F]$ in the Riemannian $\gamma$ case (i.e. $r=0, s=n$).

In Section \ref{secKdSlike} we first review (in subsection \ref{seccauchy}) some general facts on the asymptotic Cauchy problem of the $(\Lambda>0)$-vacuum field equations in all dimensions. We then establish the equivalence of a class of CKVs at $\scri$  satisfying the KID equation with a unique KV of the physical spacetime $\widetilde g$. This result had already been proven for asymptotic data in the Kerr-de Sitter-like class in \cite{Kdslike,marspeondata21}, and here we show it holds in general. 
Section \ref{secKdSlike} is concluded with subsection \ref{subsecKdS}, where we revisit the definition of the Kerr-de Sitter-like class of spacetimes in all dimensions and their main properties.

In Section \ref{secspecial} we apply Theorem \ref{theoendocords} to establish the equivalence between the set of five dimensional algebraically special spacetimes with non-degenerate optical matrix, classified in \cite{reall15}, with the Kerr-de Sitter-like class of spacetimes. We start by calculating the asymptotic initial data of the metrics in \cite{reall15}. This easily proves that such spacetimes are  contained in the Kerr-de Sitter-like class. The non-trivial part is to verify that the metrics given in \cite{reall15} exhaust the whole space of conformal classes of locally conformally flat four dimensional Riemannian metrics. The proof relies strongly on the results of Section \ref{seccoords}, because the coordinates in which the metrics in \cite{reall15} are given are adapted to spacetime null congruences, and have therefore nothing to do with (conformally) Cartesian coordinates at infinity. In fact, it is hard 
 to find an explicitly flat representative of the metric at $\scri$.

We finish this paper with some observations in Section \ref{secdiscus}. We emphasize that the application
in Section \ref{secspecial} goes beyond simply  providing an explicit example where our main theorem can be applied. The application is useful to  gain insight in the classification higher dimensional algebraically special spacetimes and point out several possible  future results.

 \section{Covariant classification of CKVs of locally conformally flat metrics}\label{seccoords}

We start with a well-known result in conformal geometry, which we prove for completeness. Recall that the Schouten tensor of a metric $g$ of dimension $n \geq 3$
is defined as
\begin{align*}
  \Sch_g = \frac{1}{n-2} \left ( \Ric_g - \frac{\Scal_g}{2(n-1)} g \right)
\end{align*}
and we denote the gradient, the Hessian and its trace (the rough Laplacian) by $\grad_g$, $\Hess_g$ and
$\Delta_g$ respectively. Scalar product with $g$ is denoted either by
$g(\cdot, \cdot)$ or $\la \cdot, \cdot\ra_g$.
\begin{lemma}
  \label{Schout}
    Let $(M,g)$ be a semi-riemannian manifold of dimension $n$ and $\Omega$ a smooth positive function \car{in $M$}. Define $\overline{g} = \frac{1}{\Omega^2} g$. Then  the respective Schouten tensors are related by
  \begin{align}
    \Sch_{\overline{g}} = \Sch_g + \frac{1}{\Omega} \Hess_g \Omega
    - \frac{1}{2 \Omega^2} | \grad_g \Omega|^2_{g} g \label{Schoutens}
  \end{align}
\end{lemma}

\begin{proof}
  The relationship between the Ricci tensors of $g$ and $\overline{g}$ is
  well-known be (e.g. \cite{kroonbook})
  \begin{align*}
    \Ric_{\overline{g}} = \Ric_g +  \frac{n-2}{\Omega} \Hess_g \Omega
    + \left ( \frac{1}{\Omega} \Delta_g \Omega -  \frac{n-1}{\Omega^2}
    |\grad_g \Omega|^2_g \right ) g.
  \end{align*}
  Taking trace with respect to $\overline{g}$ and inserting in the expression
  for $\Sch_{\overline{g}}$ the result follows at once.
\end{proof}
The following result relates the Hessians of scalar functions with respect to $g$
and with respect to $\overline{g}$. The Levi-Civita covariant derivatives of $g$, $\overline{g}$ are denoted $\nabla$, $\overline{\nabla}$
respectively.  Indices in objects constructed using geometric quantities
associated to $g$ or to $\overline{g}$
and raised and lowered with its corresponding metric. Capital Latin indices take values in $1, \cdots, n$.
\begin{lemma}
  Let $f$ be a smooth function on $M$ and $g$, $\overline{g}$ and
  $\Omega$ as before Then
  \begin{align*}
    Hess_{\overline{g}} \left ( \frac{f}{\Omega} \right )
    = \frac{1}{\Omega}  \left ( \Hess_g f
    +  f \left ( \Sch_g - \Sch_{\overline{g}} \right )
    + \left (
    \frac{f | \grad_g \Omega |^2_g}{2 \Omega^2} - \frac{1}{\Omega} \la \grad_g f,
    \grad_g \Omega \ra_g \right ) g \right )
      \end{align*}
\end{lemma}
\begin{proof}
The difference tensor $\overline{\nabla} - \nabla$ is
\begin{align}\label{eqdifftens}
  S^{A}{}_{CD} = - \delta^{A}_{C} \frac{\nabla_{D} \Omega}{\Omega}
  - \delta^{A}_{D} \frac{\nabla_{C} \Omega}{\Omega}
  + \frac{\nabla^{A} \Omega}{\Omega} g_{CD}.
\end{align}
For a covector $s_{A}$ we therefore have
\begin{align*}
  \overline{\nabla}_{A} s_{B}
  = \nabla_{A} s_{B}
  + \frac{1}{\Omega} \left ( s_{A} \nabla_{B} \Omega
  + s_{B} \nabla_{A} \Omega - s_{C} \nabla^{C} \Omega  g_{AB}
  \right )  .
\end{align*}
Applying this  to $s_{B} = \overline{\nabla}_{B} (\Omega^{-1} f )
=  \nabla_{B} (\Omega^{-1} f )$ and expanding the products in the right hand side yields
\begin{align*}
  \overline{\nabla}_{A} \overline{\nabla}_{B}
  \left ( \frac{f}{\Omega} \right )
  = \frac{1}{\Omega} \left ( \nabla_{A} \nabla_{B} f
  + f \left ( - \frac{1}{\Omega} \nabla_{A}\nabla_{B} \Omega
  + \frac{1}{\Omega^2} \nabla^{C} \Omega \nabla_{C} \Omega g_{AB} \right )
    - \frac{1}{\Omega} \nabla^{C} f \nabla_{C} \Omega g_{AB} \right )
\end{align*}
Replacing the Hessian of $\Omega$ with equation \eqref{Schoutens} the result follows.
\end{proof}

Consider a metric $\gE$ of signature $(r,s) $which is locally flat on a manifold $M$ and let $\covE$
be the corresponding Levi-Civita derivative. Let $p \in M$ and
$U_p$ a neighbourhood of $p$ where $g_E$ is flat. Since the general solution of
equation $\covE_{A} \covE_{B} f=0$, $f(p)=0$  on $U_p$ is a linear combination of
Cartesian coordinates centered at $p$, there exist precisely
$n := r + s$ linearly independent functions $\{Y^{C}\}$ satisfying
\begin{align}
  \covE_{A} \covE_{B} Y^{C} =0, \qquad Y^{A}|_p =0. \label{Hesseq}
\end{align}
We do not restrict $\{ Y^{A} \}$ to be orthogonal. More specifically, let
\begin{align}
  h^{AB} := \covE_{C}  Y^{A} \covE^{C} Y^{B}  . \label{defh}
\end{align}
It is immediate from \eqref{Hesseq} that $h^{AB}$ are constant on $U_p$. $\{ Y^{A} \}$ being linearly independent and vanishing at $p$, it is immediate that they define a coordinate system on $U_p$. It follows that $h^{AB}$ is invertible and has signature $(r,s)$.  We let $h_{AB}$ be its inverse and introduce the functions
\begin{align}
  Y^0 :=1, \qquad  Y^{n+1} := \frac{1}{2} h_{AB} Y^{A} Y^{B}
  \label{defY0Yn}
\end{align}
The following lemma provides a number of properties of the set
$\{ Y^{\alpha} \} := \{Y^0, Y^{A}, Y^{n+1} \}$.
\begin{lemma}
\label{CKVs}
  With the setup and definitions above, the following properties hold in $U_p$:
  \begin{itemize}
    \item[(i)] The functions $\{ Y^{\alpha} \}$ are linearly independent and satisfy
    \begin{align}
      \Hess_{\gE} Y^{\alpha} = \delta^{\alpha}_{n+1}
      \gE. \label{HessYall}
    \end{align}
    Moreover, the matrix
    \begin{align*}
      Q^{\alpha\beta}:= \covE_{C} Y^{\alpha} \covE^{C} Y^{\beta}
      - \delta^{\alpha}_{n+1} Y^{\beta} -
      \delta^{\beta}_{n+1} Y^{\alpha}
    \end{align*}
    is constant on $U_p$, non-degenerate and of signature $(r+1,s+1)$.
         \item[(ii)] The general solution of
    \begin{align*}
      \Hess_{\gE} f  = c \gE, \qquad c \in \mathbb{R} ,
    \end{align*}
   is a linear combination $f = c_{\alpha} Y^{\alpha}$ with  $c_{n+1}  = c$
and $c_{\alpha}$ for $\alpha \neq n+1$ arbitrary.
  \item[(iii)] For each $\alpha, \beta$ the vector field
    \begin{align}\label{eqckvs1}
      \zeta^{\alpha\beta} := Y^{\alpha} \grad_{\gE} Y^{\beta}
        - Y^{\beta} \grad_{\gE} Y^{\alpha}, \qquad \alpha <
          \beta
    \end{align}
    is a conformal Killing vector of $\gE$ satisfying
    \begin{align*}
      \pounds_{\zeta^{\alpha\beta}} \gE = 2 \left (
      Y^{\alpha} \delta^{\beta}_{n+1} -
      Y^{\beta} \delta^{\alpha}_{n+1}      \right ) \gE .
    \end{align*}
  \item[(iv)] The set $\B := \{ \zeta^{\alpha\beta}, \alpha <
    \beta \}$ is linearly independent and spans the conformal Killing algebra of $(U_p, g_E)$.
  \end{itemize}
\end{lemma}

\begin{proof}
Firstly, it is trivial that $\Hess_{\gE} Y^0 =0$ and $\Hess_{\gE} Y^A =0$ holds. From definition \eqref{defY0Yn} and since $\Hess_{\gE} Y^{A} =0$ we get
  \begin{align*}
    \covE_{C}\covE_{D} Y^{n+1} = h_{AB} \covE_{C} Y^{A} \covE_{D} Y^{B} .
  \end{align*}
  Fix any point $q \in U_p$ and define the square matrix $A^{A}{}_{B} := \covE_{B} Y^{A} |_q$. Using matrix notation  $(A)^{A}{}_{B} = A^{A}_{B}$ where the upper index denotes row and the lower index column we may write \eqref{defh} as ($t$ is the transpose)
  \begin{align*}
    (h^{\sharp}) = (A)^t  (g^{\sharp}|_q) (A) 
  \end{align*}
  where $(g^{\sharp}_q)$ and $(h^{\sharp})$ are the symmetric matrices with coefficients
  $(g^{\sharp})^{AB} = \gE^{AB}|_q$ and
  $(h^{\sharp})^{AB} = h^{AB}$. Since $(h^{\sharp})$ is invertible, so it is $(A)$ and
  \begin{align*}
  (h) = (A)^{-1} (g|_q) ((A)^{-1})^t \qquad
  \Longleftrightarrow \qquad  (A) (h) (A)^t  = ( g|_q),
\end{align*}
where $(g|_q)$ is the matrix with components
$(g|_q) _{AB} = (\gE)_{AB}|_q$. In index notation
$h_{AB} \covE_{C} Y^{A} \covE_{D} Y^{B} = (\gE)_{CD}$  at all points in $U_p$. Thus $\Hess_{\gE} Y^{n+1} = \gE$ as claimed. The constancy of $Q^{\alpha\beta}$ follows from \eqref{HessYall} because 
\begin{align*}
  \covE_{D} Q^{\alpha\beta}
  = (\covE_{D} \covE_{C} Y^{\alpha}) \covE^{C} Y^{\beta}
  + \covE_{C} Y^{\alpha} \covE_{D} \covE^{C} Y^{\beta} 
  - \delta^{\alpha}_{n+1} \covE_{D}  Y^{\beta}
  - \delta^{\beta}_{n+1} \covE_{D}  Y^{\alpha}  =0.
  \end{align*}
  Evaluating at $p$ and using that $Y^{A}|_p = Y^{n+1}|_p =0$ as well as
  \eqref{defh} yields
  \begin{align*}
    Q^{\alpha\beta} 
= \left \{ \begin{array}{ll}
-1 & \mbox{if} \quad \alpha = 0,  \beta = n+1 \\
-1 & \mbox{if} \quad \alpha = n+1,  \beta = 0 \\
                                          h^{AB} & \mbox{if} \quad  \alpha= A, \beta = B \\
                                        \end{array}     \right .
\end{align*}
  and zero otherwise because
  \begin{equation}
    Q^{A\hspace{0.5mm} n+1} = \covE_C Y^A \covE^C Y^{n+1} - Y^A = \covE_C Y^A (h_{BD} Y^B \covE^C Y^D) - Y^A = h^{AD} h_{BD} Y^B - Y^A = 0  
  \end{equation}
  and the cases $\alpha = A, \beta = 0$ also vanish trivially. 
  Since $h^{AB}$ is of signature $(r,s)$ it follows at once that $Q^{\alpha\beta}$ is of signature $(r+1,s+1)$ (at $p$ and hence everywhere) and, in particular, non-degenerate. This proves item (i).

 For item (ii), let $f$ be a function satisfying $\Hess_g f = c  g_E$ with $c$ a constant. Define the constant $f_p := f(p)$ and the function
$f_0 := f -f_p Y^0 -c Y^{n+1}$. It is clear that $f_0(p)=0$ and  $\Hess_{\gE} f_0 =0$. Thus, $f_0$ is a linear combination of $\{Y^{A}\}$. Therefore $f= c_{\alpha} Y^{\alpha}$ with $c_0 = f_p$ and $c_{n+1} = c$. Note that the constant $c_0$ can be arbitrarily chosen since for any constant $c'$ the function $f + c'$ also solves $\Hess_g f = c  g_E$.


For item (iii), the covector associated to $\zeta^{\alpha\beta}$ is
\begin{align}
  & \zeta^{\alpha\beta}_{D} = Y^{\alpha} \covE_{D} Y^{\beta}
  - Y^{\beta} \covE_{D} Y^{\alpha} 
  \\ \Longrightarrow \quad & 
  \covE_{C} \zeta^{\alpha\beta}_{D} =
  \covE_{C} Y^{\alpha} \covE_{D} Y^{\beta}
  + Y^{\alpha} \delta^{\beta}_{n+1} (\gE)_{CD}
  - \covE_{C} Y^{\beta} \covE_{D} Y^{\alpha}
  - Y^{\beta} \delta^{\alpha}_{n+1} (\gE)_{CD}
  \label{Dzeta}
\end{align}
so
\begin{align*}
  \covE_{C} \zeta^{\alpha\beta}_{D} +
  \covE_{D} \zeta^{\alpha\beta}_{C} =
   2 \left (  Y^{\alpha} \delta^{\beta}_{n+1} -
      Y^{\beta} \delta^{\alpha}_{n+1}      \right ) (\gE)_{CD}
\end{align*}
which establishes (iii). 

For the last item let $F_{\alpha\beta}$, for
  $\alpha <  \beta$, be any set of constants and define
  $F_{\beta\alpha} := - F_{\alpha\beta}$. The most general linear combination of elements in $\B$ is
    \begin{align*}
      \zeta := \sum_{\alpha < \beta}
      F_{\alpha\beta}
      \zeta^{\alpha\beta} =
      \frac{1}{2} F_{\alpha\beta}
      \zeta^{\alpha\beta}  = F_{\alpha\beta} Y^{\alpha} \grad_{\gE} Y^{\beta} .
      \end{align*}
      From item (iii) this vector satisfies
      \begin{align*}
        \pounds_{\zeta} \gE = 2 F_{\alpha\beta} Y^{\alpha} \delta^{\beta}_{n+1} \gE =
        2 F_{\alpha n+1} Y^{\alpha}  \gE = 2 (F_{0 n+1} Y^0 +
2 F_{A n+1} Y^{A} ) \gE.
      \end{align*}
      The functions $\{ Y^0, Y^A\}$ are linearly independent, so $\zeta=0$ implies $F_{\alpha n+1}=0$ and then $\zeta = 0$ reduces to
      \begin{align*}
        \zeta_{D} = F_{0 B} \covE_{D} Y^{B} +
        F_{AB} Y^{A} \covE_{D}  Y^{B} =0 .
      \end{align*}
      Evaluating at $p$ (where $Y^{A}$ vanishes) and using that 
      $(A)^{A}{}_{D} = \covE_{D} Y^{A} |_p$ is invertible, we find $F_{0 B}=0$.
      Finally, from \eqref{Dzeta}
      \begin{align*}
        \covE_{C} \zeta_{D} |_p = F_{AB} A^A{}_{C} A^{B}{}_{D} = 0
      \end{align*}
      from which $F_{AB}=0$ and the only vanishing
      linear combination in $\B$ is the zero vector. Finally, the number of independent constants $F_{\alpha\beta}$ equals to $\sum_{\beta = 1}^{n+1} \sum_{\alpha = 0}^{\beta-1} 1 = \sum_{\beta = 1}^{n+1} \beta = (n+1)(n+2)/2$, which is the dimension of the conformal Killing algebra of locally conformally flat $n$-metrics (e.g. \cite{IntroCFTschBook}). 
    \end{proof}
    Let now $g$ be a locally conformally flat metric and $\xi$ a
    conformal Killing vector of $g$. This means that at any point
    $p \in M$, there exists a neighbourhood $U_p$ of $p$ and a flat metric $g_E$
    on $U_p$ conformal to $g$. We restrict ourselves to $U_p$ in everything that follows
    and denote the covariant derivative w.r.t. $g$ as $\nabla$ and the covariant derivative with respect to $\gE$ as $\covE$.
    
    Let $\Omega: U_p \rightarrow \mathbb{R}$ be the smooth positive function satisfying $g_E = \Omega^{-2} g$.   By Lemma \ref{Schout} and $L_{\gE}=0$, this function
    satisfies the equation
        \begin{align}
          \Hess_g \Omega = \frac{1}{2 \Omega} | \grad_g \Omega|^2_g g - \Omega L_g
\label{HessOm}
        \end{align}
        We next show that we may assume that the function $\Omega$ satisfies, in addition, $\Omega_p = 1$ and $\nabla_{A} \Omega|_p =0$. The underlying reason is the freedom to conformally rescale a flat metric in such a way that it remains flat. We seek for a smooth function $\omega : U_p \rightarrow \mathbb{R}$, positive near $p$ such that  $\ov{\gE} = \omega^{-2} g_E$ is also flat. Since the curvature
        of ${\gE}$ is zero, the curvature of $\ov{\gE}$ will be zero if and only if $L_{\ovgE} =0 $ (indeed, this is immediate in dimension $n=3$ and
        as a consequence of the conformal invariance of the Weyl tensor in higher
        dimension). From Lemma \ref{Schout}, the metric $\ovgE$
        has $L_{\ovgE} =0$ if and only if $\omega$ satisfies the PDE
        \begin{align}
          \Hess_{\gE} \omega = \frac{1}{2 \omega} | \grad_{\gE} \omega|^2_{\gE} \gE \label{hessom}
                  \end{align}
%
%
                  As a consequence of of the flatness of $g_E$, the divergence of the above equation gives
                  \begin{equation}\label{lapom}
                   \covE^A \nabla_A \covE_B \omega = \covE_B \covE^A \covE_A  \omega = \covE_B \left( \frac{1}{2 \omega} | \grad_{\gE} \omega|^2_{\gE} \right) \quad \Longrightarrow\quad  \covE^A \covE_A  \omega =  \frac{1}{2 \omega} | \grad_{\gE} \omega|^2_{\gE} + K
                  \end{equation}
                  for a constant $K$. On the other hand, the trace of \eqref{hessom} is 
                \begin{equation*}
                \covE^A \covE_A  \omega =  \frac{n}{2 \omega} | \grad_{\gE} \omega|^2_{\gE},
                \end{equation*}
                which comparing with \eqref{lapom} yields that  $(2 \omega)^{-1} |\grad_{\gE} \omega|^2$ is constant on $U_p$.                            
%
        Denoting this constant by $c$ the set
        of equations to be solved is
        \begin{align}
          \Hess_{\gE} \omega = c \, \gE, \qquad
          \left .
          \left (
          |\grad_{\gE} \omega |^2_{\gE} - 2 c \omega  \right ) \right |_p =0.
          \label{equomega}
        \end{align}
        Let $\{ Y^{\alpha} \}$ be defined as before, i.e. $Y^{0} = 1$,
            \begin{align*}
              \covE_{C} \covE_{D} Y^{A} = 0, \qquad Y^{A}|_p =0, \quad
              Y^{n+1} := \frac{1}{2} h_{AB} Y^{A} Y^{B}
            \end{align*}
            with $h_{AB}$ the inverse of $h^{AB}$ defined in
            \eqref{defh}. By item (ii) in Lemma \ref{CKVs}, the general solution of the first equation in \eqref{equomega} is
        $\omega = a_0 + a_{A} Y^{A} + c Y^{n+1}$ where
        $a_0, a_A$ are arbitrary constants. Since $Y^{A} |_p =
        Y^{n+1}|_p = \covE_{A} Y^{n+1} |_p =0$,
        we get
        \begin{align*}
          0 = \left ( \covE_{C} \omega \covE^{C} \omega - 2 c \, \omega \right ) |_p = \left (
         a_{A} a_{B} \covE_{C} Y^{A}
         \covE^{C} Y^{B} - c \right ) |_p = a_{A} a_{B} h^{AB} - 2 c .
        \end{align*}
        Thus, the general solution of \eqref{equomega} is
        \begin{align*}
          \omega = a_0 Y^0 + a_{A} Y^{A} +
          \frac{1}{2} a_{A} a_{B} h^{AB} Y^{n+1}.
        \end{align*}
        Given any values  $\omega_0, \omega_{A} \in \mathbb{R}$ there exists
        a unique solution $\omega$ satisfying
        \begin{align*}
          \omega |_p = \omega_0, \qquad \covE_{C} \omega |_p = \omega_{C},
        \end{align*}
        because the algebraic problem
        \begin{align*}
          \omega|_p = \left ( a_0 Y^0 + a_{A} Y^{A} \right ) |_p =
          a_0 = \omega_0, \qquad 
          \covE_{C} \omega |_p = a_{A} \covE_{C} Y^{A} |_p = a_{A} A^{A}{}_{C} = \omega_{C}
        \end{align*}
        always admits a unique solution $\{ a_0, a_{A} \}$. Now, define
        $\ov{\Omega} := \omega \Omega$ with 
        \begin{align*}
          \omega_0 = (\Omega|_p)^{-1}, \qquad \omega_{A} = - \left . \left ( \frac{1}{\Omega^2} \nabla_{A} \Omega \right ) \right |_p. 
         \end{align*}
        It is immediate that 
        $\ov{\Omega} |_p = 1$, $\nabla_{C} \ov \Omega |_p =0$
        and  that $\ovgE = \omega^{-2} \gE = \omega^{-2} \Omega^{-2} g =
        \ov{\Omega}{}^{-2} g$ is flat (in a suitable connected
        neighbourhood of $p$ where
        $\omega$ remains positive). Clearly $\ov{\Omega}$ also satisfies
        \eqref{HessOm}.
        Dropping the overlines, we have shown:
        \begin{lemma}\label{lemmaconfac}
         For any locally conformally flat manifold $(M,g)$ and point $p \in M$ there exists a {unique} choice of conformal factor $\Omega$ which satisfies that $g_E := \Omega^2 g$ is a flat metric in a neighourhood of $p$ and 
            \begin{align}
          \Omega |_p = 1, \qquad \nabla_{A} \Omega |_p = 0.
          \label{BdryOmega}
          \end{align}
        \end{lemma}        
            \noindent From now on \emph {we make the choice of conformal factor as in Lemma \ref{lemmaconfac}}. 
                      Since $\gE^{AB}$ (the inverse of $(\gE)_{AB}$) is given by $\gE^{AB} = \Omega^2 g^{AB}$, the 
                      $(n+1)(n+2)/2$ vector fields $\zeta^{\alpha\beta}$ introduced in \eqref{eqckvs1}¨
can also be written in the form
            \begin{align*}
              \zeta^{\alpha\beta} = \Omega^2
              \left ( Y^{\alpha} \grad_{g} Y^{\beta} -
              Y^{\beta} \grad_{g} Y^{\alpha} \right ), \qquad \alpha <
              \beta.
            \end{align*} 
            We have shown in Lemma \ref{CKVs} that these vector fields span the conformal Killing algebra of $g_E$ in
            $U_p$ and hence also the conformal Killing algebra of $g$ on the same domain. We intend to compute coefficients of the decomposition
            \begin{align*}
              \xi = \sum_{\alpha < \beta} F_{\alpha \beta}
              \zeta^{\alpha\beta} =
              F_{\alpha\beta} \Omega^2 Y^{\alpha} \grad_g Y^{\beta},
              \qquad F_{\beta\alpha} = - F_{\alpha \beta}.
            \end{align*}            
            The strategy to do that is the well-known fact (see e.g.\cite{Der12}) that
            two local conformal Killing vectors $\xi_1$ and $\xi_2$ on a semi-riemannian manifold $(M,g)$, i.e. vector fields defined on a common
            open non-empty connected neighbourhood $U \subset M$ and satisfying
            \begin{align*}
              \pounds_{\xi_1} g = 2 \Psi_1 g, \qquad \pounds_{\xi_2} g = 2 \Psi_2 g
            \end{align*}
                        are the same on $U$ if and only if, at some point $p \in U$ it holds
            \begin{align}
              (\xi_1){}_{A} |_p & = (\xi_2){}_{A} |_p,  \nonumber \\
              \left ( \nabla_{[A}(\xi_1){}_{B]} \right ) |_p
              & = \left ( \nabla_{[A} (\xi_2){}_{B]} \right ) |_p, \nonumber \\
              \Psi_1 |_p & = \Psi_2 |_p, \label{data} \\
             \nabla_{A} \Psi_1 |_p & = \nabla_{A} \Psi_2 |_p  , \nonumber 
            \end{align}
            where the indices between brackets are antisymmetrized.  
Assume that we are given a conformal Killing vector $\xi$ on $(M,g)$, so we can compute
the function $\Psi_{\xi}$ defined by
$\pounds_{\xi} g = 2 \Psi_{\xi} g$. We therefore may regard 
the following quantities as known ($p$ is, as before, any chosen point in $M$)
\begin{align*}
\xizero_{A}  :=   \xi_{A} |_p, \qquad
\omegazero_{AB} := (d \bm{\xi} )_{AB} |_p 
= \left ( \nabla_{A} \xi_{B} - \nabla_{B} \xi_{A} \right ) |_p, 
\qquad \psizero := \Psi_{\xi} |_p,
\qquad \derpsizero_{A} := \nabla_{A} \Psi_{\xi} |_p
\end{align*}
where $\bm{\xi} := g (\xi, \cdot)$. Let $\zeta_F$ be defined by
\begin{align}
\zeta_F =  \frac{1}{2} F_{\alpha\beta} \zeta^{\alpha\beta} =
              F_{\alpha\beta} \Omega^2 Y^{\alpha} \grad_g Y^{\beta} =     F_{\alpha\beta} Y^{\alpha} \grad_{\gE} Y^{\beta} 
\label{defzeta} 
\end{align}
where $F_{\alpha\beta} = F_{[\alpha\beta]}$ are arbitrary constants.
By item (iii) in Lemma \ref{CKVs} we have
\begin{align}
\pounds_{\zeta_F}  g & = \pounds_{\zeta_F} ( \Omega^2 \gE) )=
  \Omega^2 \left ( 2 \frac{ \zeta_F (\Omega)}{\Omega} \gE
+ \pounds_{\zeta_F} \gE \right ) =
\Omega^2 \left ( 2 \frac{ \zeta_F (\Omega)}{\Omega} \gE
+ 2 F_{\alpha\beta} Y^{\alpha} \delta^{\beta}_{n+1} \gE \right )
\nonumber \\
& = 2 \left ( \frac{ \zeta_F (\Omega)}{\Omega} 
+ F_{\alpha\beta} Y^{\alpha} \delta^{\beta}_{n+1} \right ) g =:
2 \Psi_F \, g \label{defPsiF}
\end{align}
with the last equality defining $\Psi_F$. Let us compute the differential of 
this function
\begin{align*}
\nabla_{C} \Psi_F & = 
\nabla_{C} \left ( \frac{1}{\Omega} 
\zeta_F^{D} \nabla_{D} \Omega \right ) 
+ F_{\alpha\beta}
\nabla_{C} Y^{\alpha} \delta^{\beta}_{n+1} \\
& = - \frac{1}{\Omega^2} \zeta_F^{D} \nabla_{D} \Omega \nabla_{C} \Omega
+ \frac{1}{\Omega} \zeta_F^{D} \nabla_{C} \nabla_{D} \Omega
+ \frac{1}{\Omega} (\nabla_{C} \zeta_F^{D} ) \nabla_{D} \Omega
+ F_{\alpha\beta}
\covE_{C} Y^{\alpha} \delta^{\beta}_{n+1} \\
& =  
- \frac{1}{\Omega^2} \zeta_F^{D} \nabla_{D} \Omega \nabla_{C} \Omega
+ \zeta_F^{D} \left ( \frac{\nabla_{A} \Omega \nabla^{A} \Omega}{2 \Omega^2}   g_{CD} -
(L_g)_{CD} \right )
+ \frac{1}{\Omega} (\nabla_{C} \zeta_F^{D} ) \nabla_{D} \Omega
+ F_{\alpha\beta}
\covE_{C} Y^{\alpha} \delta^{\beta}_{n+1} 
\end{align*}
where in the third equality we inserted \eqref{HessOm}. We elaborate the third term using the difference tensor
$S = \covE - \nabla$, explicitly given in \eqref{eqdifftens},
\begin{align*}
\nabla_{C} \zeta_F^{D} & = \covE_{C} \zeta_F^{D}
- S^{D}{}_{AC} \zeta_F^{A} =
\covE_{C} \zeta_F^{D}
+ \zeta_F^{D} \frac{\nabla_{C} \Omega}{\Omega}
+ \delta^{D}_{C} \frac{1}{\Omega} \zeta_F^{A} \nabla_{A} \Omega
- \frac{\nabla^{D} \Omega}{\Omega} g_{AC} \zeta_F^{A}
\end{align*}
and insert above to get
\begin{align}
\nabla_{C} \Psi_F + \zeta_F^{D} L_g{}_{CD} & = 
 \frac{1}{\Omega^2} \zeta_F^{D} \nabla_{D} \Omega \nabla_{C} \Omega
 - 
 \frac{\nabla_{A} \Omega \nabla^{A} \Omega}{2 \Omega^2}   \zeta_F{}_{C}
+ \frac{1}{\Omega} (\covE_{C} \zeta_F^{D} ) \nabla_{D} \Omega
+ F_{\alpha\beta}
\covE_{C} Y^{\alpha} \delta^{\beta}_{n+1}  . \label{last}
\end{align}
We may now determine the coefficients $F_{\alpha\beta}$ in terms
of \eqref{data}.
  \begin{proposition}\label{propxiF}
    Let $(M,g)$ be a locally conformally flat semi-riemannian manifold
    of arbitrary signature $(r,s)$ and dimension $n \geq 3$. Fix any point $p$ in $M$ and a sufficiently small simply connected neighbourhood of $p$.
    Define $\Omega$ in $U_p$ as the unique solution of \eqref{HessOm} satisfying \eqref{BdryOmega} and let $\gE := \Omega^{-2} g$. This metric is flat in $U_p$
    and we may \car{introduce the}
    functions $\{ Y^{\alpha}\}$ and the conformal
    Killing vectors $\zeta^{\alpha\beta}$ on  $U_p$ as described above.
    
       Let $\xi$ be any conformal Killing vector on $(U_p, g)$. Define $\Psi_{\xi}$ by
    $\pounds_{\xi} g = 2 \Psi_{\xi} g$ and introduce the quantities
          \begin{align*}
\xizero_{C}  :=   \xi_{C} |_p, \qquad
\omegazero_{CD} := \frac{1}{2} (d \bm{\xi} )_{CD} |_p 
\qquad \psizero := \Psi_{\xi} |_p,
\qquad \widehat{W}_{C} := \nabla_{C} \Psi_{\xi} + \xi^{D} (L_g)_{CD}|_p
      \end{align*}
      where $L_g$ is the Schouten tensor of $g$. Then $\xi$ admits the decomposition
            \begin{align*}
        \xi = \frac{1}{2} F_{\alpha\beta} \zeta^{\alpha\beta}
        = F_{\alpha\beta} \Omega Y^{\alpha} \grad_{g} Y^{\beta}
      \end{align*}
            with $F_{\alpha\beta} = F_{[\alpha \beta]}$ given by
            \begin{align*}
              F_{0 B} = B^{C}{}¨_{B} \xizero_{C}, \quad
              F_{AB} = \widehat\omega_{CD} B^{C}{}_{A}
              B^{D}{}_{B}, \qquad
              F_{A n+1} = \widehat{W}_{C} B^{C}{}_{A}, \qquad
              F_{0 n+1} = \psizero
              \end{align*}
              where $B^{A}{}_{B}$ is the inverse of $A^{A}{}_{C} := \nabla_{C} Y^{A} |_p$. i.e.
              \begin{align*}
                A^{A}{}_{C} B^{C}{}_{B} = \delta^{A}{}_{B}.
              \end{align*}             
          \end{proposition}

  \begin{proof}
    Since $\zeta_F{}_{A} = \Omega^2 F_{\alpha\beta} Y^{\alpha}
    \nabla_{A} Y^{\beta}$ we need to impose (recall that $\covE_C Y^{n+1} |_p = 
\covE_C \Omega|_p =0$)
      \begin{align*}
        \xizero_{C} & = \zeta_{F}{}_{C} |_p =
        F_{\alpha\beta} \delta^{\alpha}_0 \covE_{C} Y^{\beta} |_p =
        F_{0 B} A^{B}{}_{C}, \\
        2 \omegazero_{CD}  & = (d \bm{\zeta_F} )_{CD} |_p
        = 2 \nabla_{[ C} \left ( \Omega^2 F_{\alpha\beta}
          Y^{\alpha} \nabla_{D]} Y^{\beta} \right ) |_p
        = 2 F_{\alpha\beta} \nabla_{[C} Y^{\alpha}
              \nabla_{D]} Y^{\beta} |_p =
            2 F_{\alpha\beta} \covE_{[C} Y^{\alpha}
                  \covE_{D]} Y^{\beta} |_p =
        2    F_{AB} A^{B}{}_{C} A^{B}{}_{D}, \\
                        \psizero = \Psi_F |_p  &= \left. \left ( \frac{1}{\Omega} \xi_F (\Omega)
            + F_{\alpha\beta } Y^{\alpha} \delta^{\beta}_{n+1}
            \right ) \right \rvert_p = F_{0 n+1}, \\
            \widehat{W}_{C} & = F_{\alpha\beta} \covE_{C} Y^{\alpha}
            \delta^{\beta}_{n+1} |_p = F_{A  n+1} A^{A}{}_{C},
                                           \end{align*}
      where in the third equality we used $\Psi_F$ as given in \eqref{defPsiF} and
      in the last equality we applied \eqref{last}. The result follows at once from these expressions.
  \end{proof}

  \begin{remark}
    \label{freedomY}
              The definition  of $\{ Y^{A} \}$ allows one to choose \underline{any} invertible matrix $A^{A}{}_{C}$ for its construction. However, if one wants $\{ Y^{A} \}$ to define a Cartesian coordinate system on $U_p$, then
              $A^{A}{}_{B}$ must be chosen so that 
              \begin{align}
                g^{CD} |_p A^{A}{}_{C} A^{B}{}_{D} = \eta^{AB}
\label{etacase}
              \end{align}
              where $\eta^{AB} = \mbox{diag} \{ \overbrace{-1,\cdots, -1}^{r}, \overbrace{1, \cdots 1}^{s}\}$. Note that the left hand side of \eqref{etacase} is by definition $h^{AB}\mid_p$ and that $h^{AB}$, which gives the entries of $g^{AB}$ in coordinates $\{ Y^{A} \}$ (see \eqref{defh}), is constant on $U_p$.
            \end{remark}

            With this proposition at hand, we can relate the coefficients $F_{\alpha\beta}$ with the determination of the conformal class to which $\xi$ belongs. For simplicity of presentation, let us assume that $(M,g)$ is compact and simply connected and let $\mbox{Conf}(M,g)$ be the conformal group, i.e.
            the collection of all conformal diffeomorphism of $(M,g)$ namely transformations $\varphi: M \rightarrow M$ satisfying $\varphi^{\star}(g)= \theta^2 g$, for some smooth positive function $\theta : M \rightarrow \mathbb{R}$. The vector field
            $\varphi^{-1}_{\star}(\xi)$ defines a conformal Killing vector of $(M,g)$ and by definition, the conformal class is the collection of all such conformal
            vector fields. It is known (see \cite{IntroCFTschBook}, \cite{blairconformal})
            that the conformal Killing algebra of $(M,g)$ is isomorphic (as a vector space) to the set of skew-symmetric
            endomorphisms $(V,\eta)$ where $V \simeq \mathbb{R}^{n+2}$ and
            $\eta$ is a metric of signature $(r+1,s+1)$. 
            The map is also an anti-isomorphism of Lie algebras. Note that the skew-symmetry here is defined with respect to the interior product $\left\langle \cdot , \cdot \right \rangle $ defined by $\eta$, namely, an endomorphism $F$ is skew-symmetric if for every pair of vectors $u,v \in \mathbb{R}$ it satisfies $\left\langle F(u),v\right \rangle = - \left\langle u,F(v)\right\rangle$. 

            Let us call the conformal Killing algebra of $(M,g)$ by $\mbox{CKill}(M,g)$, the vector space $(V,\eta)$ as $\mathbb{R}^{(r+1,s+1)}$,
            the set of skew-symmetric endomorphisms of $\mathbb{R}^{(r+1,s+1)}$
 by $\mbox{SkewEnd}(\mathbb{R}^{(r+1,s+1)})$ and denote the isomorphism above by
            \begin{align*}
              \Psi :  \mbox{CKill}(M,g) \longrightarrow
              \mbox{SkewEnd}(\mathbb{R}^{(r+1,s+1)}).
              \end{align*}
              It turns out that the conformal class of $\xi$ is
              \begin{align*}
                [\xi] = \{ \Psi^{-1} ( F_{\Lambda} ) \quad \mbox{for all} \quad 
                F_{\Lambda} = \Lambda^{-1} \Psi(\xi) \Lambda, \quad 
                \Lambda \in O(r+1,s+1) \}
              \end{align*}
                            where the orthogonal group $O(r+1,s+1)$ acts on
              $\mathbb{R}^{(r+1,s+1)}$ in a natural way.

              The construction of the map $\Psi$ can be done in several ways and is a consequence of the fact that $(M,g)$ can be isometrically embedded in the null cone of the origin in $\mathbb{R}^{r+1,s+1}$. The explicit representation of $\Psi$
              relies on a choice of a flat representative $\gE$ in the
              conformal class $[g]$ of $g$ in a sufficiently small neighbourhood $U_p$ of $p$ and a choice of Cartesian coordinates $X^{A}$ in $U_p$. To $g_E$ and $\{ X^{A}\}$ one associates an orthonormal basis $B:= \{ e_{\alpha} \}$
             of $\mathbb{R}^{r+1,s+1}$ with $e_0$ timelike. The most general conformal Killing vector
              $\xi \in \mbox{CKill}(M,g)$ restricted to $U_p$ can be
              written in terms of constants $\{ \br_A, \nu, \ar_A, \omega_{AB} =
              \omega_{[AB]}\}$ as
              \begin{align}
              \xi = \left ( \br^A + \nu X^A +  (\ar_B X^B) X^A - \frac{1}{2} (X_B X^B) \ar^A - \omega^A{}_{B} X^B \right ) \partial_{X^A} 
\label{Killvec}
               \end{align}
              where indices are raised and lowered with
              $\eta^{AB}$ (cf. remark \ref{freedomY}) and its inverse.
              Then, the endomorphism $\Psi(\xi) \in \mbox{SkewEnd} (\mathbb{R}^{r+1,s+1)})$ expressed in the basis $B$, i.e.
              $(F_{\xi})^{\alpha}{}_{\beta} e_{\alpha} = \Psi(\xi) (e_{\beta})$
              is given by  $F_{\xi}{}^{\alpha}{}_{\beta} =
\eta^{\alpha\mu} \bm{F_{\xi}}_{\mu\beta}$ 
\begin{align*}
\eta_{\alpha\mu} = \eta^{\alpha\mu} = \mbox{diag} \{ -1,1,\underbrace{-1,\cdots, -1}_{r},\underbrace{1,\cdots, 1}_{s} \}
\end{align*}
and  the matrix $(\bm{F_{\xi}})_{\mu\beta}$ is given by (the first index is row and the second column)
\begin{equation}
\bm{F_{\xi}} = 
  \begin{pmatrix}
  0 &  \nu & \ar^t - \br^t/2 \\
  -\nu & 0 & - \ar^t - \br^t/2 \\
  -\ar + \br/2 & \ar +\br/2 & -\pmb{\omega}
  \end{pmatrix}.
\end{equation}
where $\ar,\br \in \mathbb{R}^n$ are column vectors with components $\ar_A, 
\br_A$ respectively, $t$ denotes transponse and $\pmb{\omega}$ is the $n \times m$ skewsymmetric matrix with components $\omega_{AB}$.
We can now make connection to the previous decomposition. Assume for the moment that we take $A^{A}{}_{B}$ such that \eqref{etacase} holds. Then $\{Y^{A}\}$ defines
a Cartesian coordinate system of the flat metric $g_E$ in $U_p$. Moreover, since $Y^{n+1} = \frac{1}{2} \eta_{AB} Y^A Y^B$ it is straightforward to check that
setting $X^A = Y^A$, the conformal Killing vector \eqref{Killvec}
can also be written as
\begin{align*}
\xi & = F_{\alpha\beta} Y^{\alpha} \covE^A Y^{\beta} \partial_{Y^A} =
F_{\alpha\beta} Y^{\alpha} \grad_{\gE} Y^{\beta}, 
\end{align*}
with
\begin{align*}
F_{0 A} & = \br_A, \qquad F_{0 n+1} = \nu, \qquad F_{AB} = \omega_{AB}
\qquad F_{C, n+1} = a_C
\end{align*}
In order to match the two constructions we need to introduce a vector
space of dimension $n+2$ and a metric of signature $(r+1,s+1)$. Define
$\hat{V} := \mbox{span}\{ Y^{\alpha} \}$ and endow this space with the scalar
product $Q$ defined by
\begin{align*}
  Q (Y^{\alpha}, Y^{\beta} ) := Q^{\alpha\beta}
\end{align*}
where the constants $Q^{\alpha\beta}$ are defined in item (i) of Lemma
\ref{CKVs}. This metric is well-defined (i.e. independent of the choice of
functions $Y^{A}$) because under a $GL(n)$ transformation\footnote{Recall that $\{Y^A\}$ spans the solution space of \eqref{Hesseq}. Thus any other set of linearly independent solutions $\{Y'^A\}$ must necessarily differ from $\{Y^A\}$ by a $GL(n)$ transformation.} $M$
\begin{align*}
  Y'{}^{A} = M^{A}{}_{B} Y^{B}
\end{align*}
we have $Y'^{0} = Y^0$ (obvious) and $Y'{}^{n+1} = Y^{n+1}$ because
\begin{align*}
    Y'{}^{n+1} = \frac{1}{2} h'_{AB} Y'{}^A Y'{}^B =
    \frac{1}{2} h'_{AB} M^A{}_C M^{B}_C Y^C Y^D = \frac{1}{2} h_{CD} Y^A
    Y^D = Y^{n+1}
\end{align*}
and the last equality follows from definition \eqref{defh} (and its corresponding prime)
together with
\begin{align*}
  A'{}^{A}_{B} := \covE_B Y'{}^A |_p = M^A{}_C A^{C}{}_B.
  \end{align*}
    Hence, from the definition of $Q^{\alpha\beta}$ we have
    \begin{align*}
      Q'{}^{\alpha\beta} = Q^{\mu\nu} M^{\alpha}{}_{\mu} M^{\beta}{}_{\nu},
    \end{align*}
    with $M^{0}{}_{0} = M^{n+1}{}_{n+1}=1$, $M^{\alpha}{}_{\beta} = M^A{}_B$
    for $\alpha =A$, $\beta=B$ and the rest are zero.

    In the case that $\{Y^{A}\}$ are Cartesian coordinates (i.e. when $h^{AB} = \eta^{AB}$), then we can construct an orthonormal basis of $\hat{V}$ by introducing
    \begin{align*}
      E^0 := Y^0 + \frac{1}{2} Y^{n+1}, \qquad
      E^1 :=  Y^0 - \frac{1}{2}Y^{n+1}, \qquad
      E^{A+1} := Y^{A} .
     \end{align*}
     The endomorphism of $\hat{V}$ defined by
     \begin{align*}
     F_{\xi}(Y^{\alpha}) := Q^{\alpha\beta} F_{\beta\mu} Y^{\mu}
     \end{align*}
     is identical to the endomorphism $\Psi(\xi)$ if we identify
     $V = \hat{V}$ and the basis vectors $e_{\alpha} = \eta_{\alpha\beta} E^{\beta}$.

     The classification of the endomorphism $\Psi(\xi)$ up to conjugacy class is obviously independent of the choice of basis. Thus, once we have established the equivalence of $\Psi(\xi)$ and $F_{\xi}$ we may use {\it any} basis $\{Y^{\alpha}\}$, not necessary orthogonal. From the point of view of the original space $(M,g)$ a natural choice is the (non-orthogonal) basis defined by
     \begin{align*}
       \covE_A Y^B |_p = \delta^B{}_A .
     \end{align*}
     In such a basis, the expression of $F_{\alpha\beta}$ is simplest, while the expression of $Q^{\alpha\beta}$ is just
     \begin{align*}
       Q^{0,n+1} = Q^{n+1,0} = -1, \qquad Q^{AB} = g^{AB} |_p .
     \end{align*}
     We summarize this result in the following theorem:
     \begin{theorem}\label{theoendocords}
\label{endomorfism}
       Let $(M,g)$ be a locally conformally flat semi-riemannian space of arbitrary
       signature $(r,s)$ and dimension $n \geq 3$. Fix a point $p \in M$
       and a local conformal Killing vector $\xi$ defined in a sufficiently small open neighbourhood $U_p$ of $p$. Let $\Psi_{\xi} : U_p \rightarrow \mathbb{R}$
       be defined by $\pounds_{\xi} g = 2 \Psi_{\xi}$.

       Then the conformal class of $\xi$ is determined by the
       conjugacy class under $O(r+1,s+1)$
       of the skew-symmetric endomorphism $F_{\xi}$ on $(\mathbb{R}^{r+1,s+1}, Q)$ defined by
       $F_{\xi} (v^{\alpha}) = Q^{\alpha\beta} F_{\beta\mu} v^{\mu}$ where
       $v^{\alpha}$ is a basis of $\mathbb{R}^{r+1,s+1}$
       with non-zero scalar products
       \begin{align*}
         Q^{\alpha\beta} := Q(v^{\alpha}, v^{\beta}) = \left \{ \begin{array}{ll}
           -1 & \mbox{if} \quad  \alpha =0, \beta = n+1 \quad \mbox{or} \quad
           \alpha =n+1, \beta = 0 \\
           g^{AB} |_p &  \mbox{if} \quad \alpha = A, \beta = B \\
           0 &  \mbox{rest of terms} 
         \end{array}
         \right . 
       \end{align*}
       and the coefficients $F_{\beta\mu} = F_{[\beta\mu]}$ are given by
       \begin{align*}
         F_{0, A} = \xi_A |_p, \qquad
         F_{AB} = \nabla_{[A} \xi_{B]} |_p, \qquad
           F_{0, n+1} = \Psi_{\xi} |_p, \qquad
           F_{A, n+1} = ( \nabla_A\Psi_{\xi} + \xi^B (L_g)_{AB} ) |_p
       \end{align*}
       where $L_g$ is the Schouten tensor of $g$ and $\nabla$ its covariant derivative.

       \end{theorem}
     
     \subsection{The Riemannian case ($r=0, s=n$)}\label{secriem}
The method of classification of CVKs employed in \cite{marspeondata21,marspeonKSKdS21} requires to find explicitly a flat representative $g_E$ in the class of locally conformally flat metrics and also Cartesian coordinates for $g_E$. However, this may be a very hard task.  Theorem \ref{theoendocords} improves the classification method in  \cite{marspeondata21,marspeonKSKdS21} as it allows to obtain the conformal class of a CKV with {\it independence on the coordinates and the representative} of the class of conformally flat metrics. We shall provide an interesting application of this result in Section \ref{secspecial}. For that, we now introduce the explicit classification of CKVs in conformally flat Riemannian metrics.

     From the discussion above it follows that, for locally conformally flat Riemannian  $n$-metrics, the classification of conformal \car{classes} of CKVs is equivalent to the classification of $\mbox{SkewEnd}(\mathbb{R}^{(1,n+1)})$ up to $O(1,n+1)$ transformations. In order to uniquely characterize the conjugacy class $[F] = \{F_\Lambda \in \mbox{SkewEnd}(\mathbb{R}^{(1,n+1)}) \mid F_\Lambda = \Lambda F \Lambda^{-1},\quad \Lambda \in O(1,n+1)\}$ one needs to find a sufficient number of $ O(1,n+1)$-invariant quantities. A possibility \cite{marspeon21} is to give the eigenvalues\footnote{It is preferable to use the eigenvalues of $F^2$ because they are real and they are in one-to-one correspondence with the (complex) eigenvalues of $F$.} of $F^2$ together with the causal character of $\ker F$ (see also \cite{Kdslike} for an alternative classification in terms of the traces of even powers of $F$ and its matrix rank). Observe \cite{marspeon21} that, as a consequence of the skew-symmetry of $F$, all eigenvalues of $F^2$ are at least of double multiplicity and there is always a vanishing eigenvalue if $n$ is odd. Hence, it turns out \cite{marspeon21} to be sufficient to determine the roots of the following polynomial
\begin{equation}\label{defQF2}
\mathcal{Q}_{F^2}(x) := \lr{\mathcal{P}_{F^2}(-x)}^{1/2} \quad \mbox{(if $n$ even)}, \quad\quad \mathcal{Q}_{F^2}(x) := \lr{\frac{\mathcal{P}_{F^2}(-x)}{x}}^{1/2} \quad \mbox{(if $n$ odd)},
\end{equation}
where $\mathcal{P}_{F^2}(-x)$ refers to the characteristic polynomial of $-F^2$. Counting multiplicity, $\mathcal{Q}_{F^2}(x)$ has $p$ roots, where $p$ is the natural number related to the dimension $n$ by 
\begin{equation}\label{eqpq}
 p := \left\lfloor\frac{n+1}{2}\right\rfloor -1,
\end{equation}
\car{being $\lfloor x \rfloor \in \mathbb{Z} $ the floor function for all $x \in \mathbb{R}$, i.e. the largest integer which is equal or less than $x$.}
Then, the classification result of equivalence classes of $F$ is given by the following Proposition:
\begin{proposition}[\hspace{-0.025cm}\cite{marspeon21}]\label{defgammamu}
 Let $\mathrm{Roots}\lr{\mathcal{Q}_{F^2}}$ denote the set of roots of $\mathcal{Q}_{F^2}(x)$ repeated as many times as their multiplicity and arranged as follows:
 \begin{enumerate}
  \item[a)] If $n$ odd, $\lrbrace{\sigma; \mu_1^2, \cdots, \mu_{p}^2} := \mathrm{Roots}\lr{\mathcal{Q}_{F^2}}$ sorted by $\sigma \geq \mu_1^2\geq \cdots \geq \mu_{p}^2$ if $\ker F$ is timelike, where in this case necessarily $\sigma > 0$. Otherwise $\mu_1^2\geq \cdots \geq \mu_{p}^2\geq 0 \geq \sigma$. 
  \item[b)] If $n$ even, 
  $\lrbrace{-\mu_t^2, \mu_s^2; \mu_1^2, \cdots, \mu_{p}^2} := \mathrm{Roots}\lr{\mathcal{Q}_{F^2}}$ sorted by $\mu_1^2\geq \cdots \geq \mu_{p}^2\geq\mu_s^2 = -\mu_t^2 = 0$ if $\ker F$ is null. Otherwise  $\mu_s^2 \geq \mu_1^2\geq \cdots \geq \mu_{p}^2\geq 0 \geq  -\mu_t^2$, where either $\mu_s^2$ or $\mu_t^2$ are non-zero.  
 \end{enumerate}
 Then the parameters $\lrbrace{\sigma; \mu_1^2, \cdots, \mu_{p}^2}$ for $n$ odd and $\lrbrace{-\mu_t^2, \mu_s^2; \mu_1^2, \cdots, \mu_{p}^2}$  for $n$ even determine uniquely the class of $F$ up to $O(1,n+1)$ transformations and hence also the class of $\xi := \Psi^{-1}(F)$ up to conformal transformations.  
\end{proposition}

\begin{remark}\label{remarknullker}
 It holds in general that  $\mathcal{Q}_{F^2}$ has at most one negative root \cite{marspeon21}. When $n$ is even and $\ker F$ is null, then  $\mathcal{Q}_{F^2}$ has a root at least double at zero and no negative roots  \cite{marspeon21}. These are necessary facts that follow from Proposition \ref{defgammamu}.
\end{remark}

 Note that endomorphisms with equal roots of $\mathcal{Q}_{F^2}$ (hence equal eigenvalues with same multiplicities) can belong to different conformal classes. The idea in \cite{marspeon21} is to introduce an additional invariant, namely the causal character of $\ker F$, to remove this ambiguity by defining $\mathrm{Roots}\lr{\mathcal{Q}_{F^2}}$, whose elements are sorted depending on $\ker F$. This gives a well-defined parametrization of the space of conformal classes, which will be key in Section \ref{secspecial}.

  \section{The asymptotic Cauchy problem and the Kerr-de Sitter-like class}\label{secKdSlike}

    In this section we review the basics on the asymptotic Cauchy problem for $(\Lambda > 0)$-vacuum, $(n+1)$-dimensional spacetimes and introduce the definition and properties of the Kerr-de Sitter-like class of spacetimes, in four dimensions \cite{Kdslike,KdSnullinfty}, as well as its extension to higher dimensions  \cite{marspeondata21,marspeonKSKdS21}. The results in this section are not new, but they will be needed in Section \ref{secspecial}. The following discussion is meant to summarize these results in order to make the paper self-contained. 
   


\subsection{Asymptotic Cauchy problem with $\Lambda >0$}\label{seccauchy}

{As we already mentioned in the introduction}, in some situations, $(n+1)$-dimensional, $(\Lambda>0)$-vacuum spacetimes $(\widetilde{{M}},\widetilde g)$ admitting a conformal extension $({M},g := \Omega^2 \widetilde g)$ can be characterized in a neighbourhood of $\scri$ by asymptotic initial data (i.e. data prescribed at $\scri$). This is true in general for $n=3$ by the classical results by Friedrich \cite{friedrich81,friedrich81bis,Fried86geodcomp}. In higher dimensions, the asymptotic characterization results stem from the Fefferman and Graham formalism \cite{FeffGrah85,ambientmetric}, and hold in general for $(n+1)$ even dimensions \cite{Anderson2005,andersonchrusciel05,Kaminski21} or when the asymptotic initial data are analytic \cite{kichenassamy03}. 

In what follows we restrict to $(n+1)$-dimensional (with $n \geq 3$), $(\Lambda>0)$-vacuum metrics which admit a locally conformally flat $\scri$. In this case, the asymptotic initial data 
\cite{marspeondata21} is a conformally flat Riemannian $n$-manifold $(\Sigma,\gamma)$, which prescribes the geometry of $\scri$ via an isometric embedding $\iota: \Sigma \hookrightarrow \scri \subset M$, together with a transverse and traceless tensor, or TT tensor,  $D$ which prescribes the electric part of the rescaled Weyl tensor. Namely
\begin{equation}\label{rescweyl}
 D := \iota^\star\lr{ |\grad_g \Omega|^{-2} \Omega^{2-n} C(\grad_g \Omega,\cdot, \grad_g \Omega,\cdot)}
\end{equation}
where $C$ is the $4$-covariant Weyl tensor of $({M},g)$. The conformal flatness of $\scri$ is not required in the $n=3$ case for \eqref{rescweyl} to hold true,
but it is indeed needed for $n>3$ \cite{marspeondata21}. Actually, the (smooth) extendability of $\Omega^{2-n} C(\grad_g \Omega,\cdot, \grad_g \Omega,\cdot)$ to $\scri$ is a non-trivial result (cf. \cite{marspeondata21,Holl05}) for $n>3$, which relies strongly on the assumption of local conformal flatness of $\scri$. For $n = 3$ the extendability of $\Omega^{-1} C$ to a generic $\scri$ is a consequence of the Weyl tensor vanishing identically in three dimensions.

In addition, the characterization of spacetimes in terms of asymptotic initial data is independent of the conformal factor $\Omega$. As a consequence the data have the following conformal equivalence 
\begin{equation}\label{equivdatan}
 (\Sigma,\gamma,D) \simeq (\Sigma,\omega^2 \gamma, \omega^{2-n} D),\quad\quad  \omega \in C^\infty(\Sigma,\mathbb{R}^+),
\end{equation}
in the sense that any pair of data correspond to the same physical spacetime $(\widetilde{{M}},\widetilde g)$ if and only if they are related by \eqref{equivdatan}.

Now, let $(\Sigma,\gamma,D)$ be asymptotic data for $(\widetilde M, \widetilde g)$. Then, for a CKV $\xi$ of $\gamma$, the following KID equation 
\begin{equation}\label{kideq}
\mathcal{L}_\xi D + \frac{n-2}{n} (\mathrm{div}_\gamma \xi) D = 0
\end{equation}
is proven to be a necessary and sufficient condition for $\widetilde g$ to admit a KV $X$ such that $\xi = X \mid_\scri$, in general if $n=3$\cite{KIDPaetz} and assuming that $(\gamma,D)$ are analytic for $n > 3$ \cite{marspeondata21} (there is no proof yet for the general $n>3$ case, but we believe that \eqref{kideq} asymptotically characterizes symmetries in general.) 
It is a matter of direct computation to show that if $(\Sigma,\gamma,D; \xi)$ satisfies the KID equation, \car{so does} $(\Sigma,\omega^2 \gamma,\omega^{2-n} D; \xi)$. Then, for any element $\varphi \in \conf{\Sigma,\gamma}$, the following equivalences are ready
\begin{equation}\label{equivKID}
 (\Sigma,\gamma,D ; \xi) \simeq (\Sigma,\varphi^\star(\gamma),\varphi^\star(D) ; \varphi_\star^{-1}(\xi)) = (\Sigma,\omega^2 \gamma, \omega^{2-n} D; \varphi_\star^{-1}(\xi)) \simeq (\Sigma,\gamma,D ; \varphi_\star^{-1}(\xi)),
\end{equation}
the first one arising from the fact that $\varphi$ is a diffeomorphism and the last one from \eqref{equivdatan}. Therefore, a particular KV of the bulk spacetime is not associated to a single CKV $\xi$ satisfying \eqref{kideq}, but its whole conformal class $[\xi]$ of CKVs satisfying \eqref{kideq} (here it is crucial that each representative of $[\xi]$ is a solution of \eqref{kideq}.).

\subsection{The Kerr-de Sitter-like class}\label{subsecKdS}

In four spacetime dimensions (i.e. $n=3$), the vanishing of the Mars-Simon tensor for a particular KV \cite{mars99,simon84} characterizes Kerr-de Sitter and related spacetimes \cite{marsseno15}. The class of four-dimensional spacetimes which are $(\Lambda>0)$-vacuum, conformally extendable and admitting a KV $X$ whose Mars-Simon tensor vanishes defines the so-called Kerr-de Sitter-like class of spacetimes. In the locally conformally flat $\scri$ case\footnote{The non-conformally flat $\scri$ cases were also studied in \cite{KdSnullinfty}, but the results are not needed for our purposes here. } \cite{Kdslike}, the asymptotic data characterizing this class is a conformally flat Riemannian $3$-manifold $(\Sigma,\gamma)$ and a TT tensor $D$ of the form 
\begin{align}
D = \kappa D_\xi,\quad\mbox{where}\quad  D_\xi := \frac{1}{|\xi|_\gamma^5}\lr{\bm\xi \otimes \bm\xi - \frac{|\xi|_\gamma^2}{3}\gamma},
\end{align}
$\kappa$ is a real constant, $\xi$ is a CKV of $\gamma$ satisfying $X\mid_\scri =\xi$ (thus also \eqref{kideq}) and $\bm \xi := \gamma(\xi,\cdot)$.
%
Using the results discussed above on asymptotic characterization of $(n+1)$-dimensional spacetimes one can extend, by means of asymptotic data, the definition of the Kerr-de Sitter-like class in the conformally flat $\scri$ case to all dimensions \cite{marspeonKSKdS21}. Namely:
\begin{definition}\label{defKdSlike}
 The $(n+1)$-dimensional Kerr-de Sitter-like class with conformally flat $\scri$ is defined as the set of $(\Lambda>0)$-vacuum spacetimes which admit a conformally flat $\scri$ and such that 
\begin{align}\label{eqDKdS}
 D = \kappa D_\xi,\quad\mbox{where}\quad  D_\xi := \frac{1}{|\xi|_\gamma^{n+2}}\lr{\bm\xi \otimes \bm\xi - \frac{|\xi|_\gamma^2}{n}\gamma},
\end{align}
$\kappa$ is a real constant and $\xi$ is a CKV of the (conformally flat) metric $\gamma$ at $\scri$.
\end{definition}
We remark that this is not an ad-hoc definition, but it naturally follows after checking \cite{marspeondata21} that the asymptotic data of the Gibbons et al. definition of the  Kerr-de Sitter metrics \cite{Gibbons2005} consist of a conformally flat manifold $(\Sigma,\gamma)$ and a tensor $D$ of the form \eqref{eqDKdS}, for a particular choice of $\xi$. Allowing $\xi$ to be an arbitrary CKV keeps the traceless and transverse property of $D_{\xi}$, and thus provides a natural generalization of the data, which in turn generalizes the definition of Kerr-de Sitter-like class to higher dimensions. Moreover, in all cases $\xi$ satisfies the KID equation \eqref{kideq}.

From the data equivalence \eqref{equivdatan} it follows \cite{marspeonKSKdS21} that the data sets in the Kerr-de Sitter-like class satisfy the following equivalence property
\begin{equation}\label{equivKdS}
 (\Sigma,\gamma,\kappa D_\xi) \simeq (\Sigma,\gamma,\kappa D_{\xi'}) \iff \xi' = \varphi(\xi),\quad\mbox{for some}\quad\varphi\in\mathrm{Conf}(\Sigma,\gamma).
\end{equation}
This has more serious consequences than simply the fact that the conformal class $[\xi]$ characterizes a unique KV (cf. \eqref{equivKID}). Due to the role that $[\xi]$ plays in the construction of the data in the Kerr-de Sitter-like class, equivalence \eqref{equivKdS} actually means
that the metrics in the Kerr-de Sitter-like class (with conformally flat $\scri$) are in one-to-one correspondence with the conformal classes of CKVs of locally conformally flat metrics. Hence, the moduli space of metrics in this class is respresented by the space of parameters in Proposition \ref{defgammamu}. Moreover, the quotient topology naturally defined in the space of conformal classes is inherited by the  space of metrics in the Kerr-de Sitter-like class. The main result in \cite{marspeonKSKdS21} exploits 
this fact to provide an explicit reconstruction of all the metrics in the Kerr-de Sitter-like class as limits of Kerr-de Sitter or an analytic extension thereof. Moreover they are proven to exhaust the set of Kerr-Schild metrics on a locally de Sitter background with an additional decay condition: 


\begin{theorem}{\cite{marspeonKSKdS21}}\label{theoKS}
 Let $(\widetilde{\mathcal{M}},\widetilde g)$ be a $(\Lambda>0)$-vacuum $(n+1)$-dimensional spacetime, such that $\widetilde g$ admits the Kerr-Schild form on a locally de Sitter background, namely
 \begin{equation}\label{KSchild}
  \widetilde g = \widetilde g_{dS} + \mathcal{H} \bm k \otimes \bm k
 \end{equation}
where $\widetilde g_{dS}$ is locally isometric to de Sitter, $\mathcal{H}$ is a smooth function on $\widetilde{\mathcal{M}}$ and $\bm k$ a null one-form (w.r.t. both $\widetilde g$ and $\widetilde g_{dS}$). Additionaly, assume that $\widetilde g$ admits a smooth conformal extension such that $\Omega^2 \mathcal{H} \bm k \otimes \bm k = O(\Omega)$. Then and only then $\widetilde g$ belongs to the Kerr-de Sitter-like class with locally conformally flat $\scri$. 
\end{theorem}

%

     \section{An application: Classification of algebraically special
       $(\Lambda>0)$-vacuum solution in five dimensions}\label{secspecial}

     In \cite{reall15}, the problem of determining the
     most general algebraically special spacetime in five dimensions that solves
     the vacuum Einstein field equations
     \begin{align*}
       R_{\alpha\beta} = 4 \lambda g_{\alpha\beta}, \qquad \lambda \in \mathbb{R} ,
     \end{align*}
     is addressed\footnote{In \cite{Freitas16}, $\Lambda = 4 \lambda$ is used instead. We prefer $\lambda$, which matches directly
the notation used \cite{marspeondata21,marspeonKSKdS21}}. Algebraically special means that the spacetime $(M,g)$
     admits a Multiple Weyl Aligned Null Direction (WAND) $\ell$. A multiple
     WAND is a non-identically vanishing  null vector field
     satisfying \cite{Ortaggio12}
     \begin{align*}
       \ell^b \ell_{[e} C_{a]b[cd} \ell_{f]} =0
     \end{align*}
     where $C_{abcd}$ is the Weyl tensor of $(M,g)$. Multiple WANDs are always geodesic,
     $\nabla_{\ell} \ell \propto \ell$. Admitting a multiple WAND is equivalent to the algebraic classification of the Weyl tensor, as extended by Coley {\it et. al.} at $d$-dimensions
     \cite{coley04}, being of type II or more special.

     The problem in \cite{reall15} is solved under the additional hypothesis that the so-called optical matrix of $\ell$ is non-degenerate\footnote{The problem is also solved in the degenerate case in \cite{Freitas16} and references therein. We restrict to the non-degenerate case.}. The optical matrix
     $\rho$ encodes the kinematical properties of the congruence of null geodesics
     defined by $\ell$. Geometrically,  it is
     a $(0,2)$-tensor defined on the quotient vector space at each point
     $p \in M$ of  equivalence classes
     $\ell^{\perp} / \sim$, where two vectors $X,Y \in T_p M$
     orthogonal to $\ell$
     are related by $\sim$ if and only if $X-Y$ is proportional to $\ell$. Its definition is
          \begin{align*}
       \rho (\ov{X}, \ov{Y} ) :=g(X, \nabla_{Y} \ell ), \qquad
       X \in \ov{X}, Y \in \ov{Y}.
     \end{align*}
     and one checks at once that $\rho$ is well-defined, i.e. independent of the
     choice of representative $X\in \ov{X}$, $Y \in \ov{Y}$.

     Since our interest here lies in the case $\lambda >0$ we quote the results in \cite{reall15} restricted to this situation. Specifically,
     the main result in \cite{reall15} states that under the
     above conditions and $\lambda >0$
     the most general solution of the  Einstein field equations
     belongs to one of three families of metrics, classified according to the eigenvalues of $\rho$. Namely, the eigenvalues of $\rho$ can be written in terms of $r$ and $\chi$, being the former the affine parameter along the null geodesics of $\ell$ and the latter a constant function along the same congruence of geodesics. The three cases arise depending on whether $\chi$
     is not everyhwere constant or, if constant, whether this constant is zero or not. 
     \begin{itemize}
     \item[Case 1] ($\chi \neq 0, d \chi \neq 0$). The functions $\{ r,\chi\}$ are completed to local coordinates $\{r,u,\chi,x,y\}$
       and in terms of real constants $A_0, \mu_0, C_0, E_0$ such that
       \begin{align*}
         P(\chi):= C_0 - \frac{E_0^2}{\chi^2} - 2 A_0 \chi^2 - \lambda \chi^4
       \end{align*}
       is positive in some interval $I \subset \mathbb{R} \setminus \{ 0 \}$ and with $\chi$ taking values on $I$ the metric is 
       \begin{align*}
         \widetilde{g}_1 = - 2 \bm{\sigma}^1 dr
         + \frac{r^2 + \chi^2}{P(\chi)} d\chi^2
         + h_{ij} \bm{\sigma}^i \bm{\sigma}^j,
       \end{align*}
       where  the one-forms $\bm{\sigma}^i$ are
            $\bm{\sigma}^1 := du + \chi^2 dy$, $\bm{\sigma}^2 := dx - \frac{E_0}{\chi^2} dy$,
            $\bm{\sigma}^3:=dy$
                 and $h_{ij}$ is the  matrix
       \begin{align*}
         h_{ij} = \left ( \begin{array}{lll}
                            -2 H(r) & E_0 & - P(\chi) \\
                            E_0 & r^2 \chi^2 & 0 \\
                            - P(\chi) & 0 & (r^2 + \chi^2) P(\chi)
                          \end{array}
                                            \right ), \qquad
                                            H(r):= A_0 - \frac{\lambda}{2}
                                            \left ( r^2 - \chi^2 \right )
                                            - \frac{\mu_0}{2 (r^2 + \chi^2)}    .  \end{align*}
     \item[Case 2]  ($\chi \neq 0, d \chi = 0$). Here $\chi$ acts as a parameter and there exist local coordinates $\{r,u,x, \theta,\phi\}$
       and real constants $E_0, \mu_0$ such that, with $\bm{\tau} :=
       du + 
       \frac{2 \chi}{F_0} \cos \theta d \phi$, the metric reads
       \begin{align*}
         \widetilde g_2 & = -2 \bm{\tau} dr - G(r) \bm{\tau}^2 + r^2 \left ( dx + \frac{E_0}{\chi^3}
         \left ( 1+ \frac{\chi^2}{r^2} \right ) \tau  \right )^2 + \frac{r^2+ \chi^2}{F_0} \left ( d \theta^2 + \sin^2 \theta d \phi^2 \right ), \\
         G(r) & := \frac{E_0^2}{\chi^4} \left ( 1
         + \frac{\chi^2}{r^2} \right ) - \lambda \left ( r^2 + \chi^2 \right )-
         \frac{\mu_0}{r^2 + \chi^2}, \qquad F_0 :=
         4 \left ( \lambda \chi^2  + \frac{E_0^2}{\chi^4} \right ) .
       \end{align*}
     \item[Case 3]  ($\chi = 0, d \chi = 0$). The metric is
       \begin{align*}
       \widetilde  g_3 = - \left ( \kappa - \frac{\mu_0}{r^2} - \lambda r^2 \right ) du^2 - 2 du dr + r^2 h
     \end{align*}
     where $\mu_0 \in \mathbb{R}$ and $h$ is a Riemannian 3-dimensional metric of constant curvature $\kappa \in \{ -1, 0, 1\}$.
     \end{itemize}

     In all cases, the metrics admit a smooth conformal compactification. Indeed, the metric $g_i = \frac{1}{r^2} \tilde{g}_i$
       $i=1,2,3$ followed by the change of variable $\Omega=1/r$ yields a metric that is smooth in $\Omega$ and extends as a Lorentzian metric to $\scri:=\{ \Omega=0\}$.
       The corresponding metrics at $\scri$, denoted by $\gamma_i$ take the form
       \begin{align*}
       \gamma_1 &=  \frac{d\chi^2}{P(\chi)} + \lambda \left (du + \chi^2 dy \right )^2
                    + \chi^2 \left (dx - \frac{E_0}{\chi^2} dy \right )^2 + P(\chi) dy^2, \\
                    \gamma_2 &= \lambda \bm{\tau}^2 + \left ( dx + \frac{E_0}{\chi^3} \bm{\tau} \right )^2 + \frac{1}{F_0} \left ( d\theta^2 + \sin^2 \theta d
                    \phi^2 \right )
                    \\
               \gamma_3 & = \lambda du^2 + h.
       \end{align*}
                    By direct computation, one can check that the Weyl tensor of each one of these  metrics is identically zero, so $\gamma_i$
                     are locally conformally flat. As discussed in Section \ref{secKdSlike}, in the locally conformally flat $\scri$ case the rescaled Weyl tensor is smoothly extendable to $\scri$ and prescribes the TT tensor $D$ in the asymptotic initial data (cf. \eqref{rescweyl}), after a suitable idenfication via the isometric embedding $\iota: (\Sigma,\gamma ) \hookrightarrow (M,g)$ such that $\iota(\Sigma) = \scri$.        

In the present case, $\scri$ is simply $\{ \Omega =0 \}$  and the embedding is trivial in these adapted coordinates, so 
it is a matter of direct computation to determine
the tensor $D$ via formula \eqref{rescweyl}. It turns out that in all three cases, this tensor
takes the following form
\begin{align*}
D_{AB} =  \frac{ 4 \mu_0 \lambda^2}{| \xi|_{\gamma}^6} 
\left ( \xi_{A} \xi_B - \frac{| \xi|_{\gamma}^2}{4} \gamma_{AB} \right )
\end{align*}
 and the vector $\xi$ are given in each case by the following expressions 
\begin{align*}
  \mbox{(Case 1):} \qquad  & \xi_1 = \partial_u, \\
  \mbox{(Case 2):} \qquad  & \xi_2 =
  \partial_u - \frac{E_0}{\chi^3} \partial_x, \\
  \mbox{(Case 2):} \qquad  & \xi_3 =  \partial_u. \\
\end{align*}
Moreover it is also a matter of direct computation to check that each vector field $\xi_i$ is a conformal Killing vector of the corresponding metric $\gamma_i$. In other words:  \emph{all the  three cases belong to the Kerr-de Sitter-like class} (cf. Definition \ref{defKdSlike}). Now a natural question is to ask whether they  are a subset within or they span the whole Kerr-de Sitter-like class. To address this question we use the fact that the metrics in this class are in one-to-one correspondence (cf. \eqref{equivdatan}) with the conformal classes of the CKVs determining $D$. Hence, we must check that all possible conformal classes $\{[\xi_i]\}$ for the admissible values of the parameters in $\tilde g_i$ cover the space of conformal classes given in Proposition \ref{defgammamu}. 

{Observe that the parameter $\mu_0$ always appears as a scaling constant in $D_{AB}$. This means that it can be set to\footnote{By considering $\sign (0) = 0 $ we may include the case $\mu _0 = 0$, which corresponds to de Sitter-spacetime.} $\epsilon = \sign(\mu_0)$ by suitably absorbing its norm into $\xi$. Namely, by defining  $\xi':= |\mu_0|^{-1/2} \xi$ it follows
\begin{align}\label{escfreedom}
D_{AB} =  \frac{ 4 \mu_0 \lambda^2}{| \xi|_{\gamma}^6} 
\left ( \xi_{A} \xi_B - \frac{| \xi|_{\gamma}^2}{4} \gamma_{AB} \right ) = \frac{ 4 \epsilon \lambda^2}{| \xi'|_{\gamma}^6} 
\left ( \xi'_{A} \xi'_B - \frac{| \xi'|_{\gamma}^2}{4} \gamma_{AB} \right ).
\end{align}
This scale freedom will be relevant to prove that $\{[\xi_i]\}$ covers all possible conformal classes of four-dimensional conformally flat metrics. }

In order to determine the conformal class of $\xi_i$, we use the results of Section \ref{seccoords}. One simply needs to fix any point $p \in \scri$, compute the
quantities associated to $\xi_i$ that appear in Theorem \ref{endomorfism} and construct the endomorphism. The result, expressed in the basis $v^{\alpha}$ of
Theorem \ref{endomorfism} is
\begin{align*}
  \mbox{Case 1}: \qquad &  F_{\xi_1} = \left ( \begin{array}{cccccc}
0 & 0 & - \lambda
\left ( A_0 + \frac{\lambda}{2} \chi_0^2 \right ) & E_0 \lambda & 
\lambda \left ( A_0 \chi_0^2 - C_0 + \frac{\lambda}{2} \chi_0^4 \right ) & 0 \\
0 & 0 & 0 & 0 & \lambda \chi_0 P(\chi_0)  & 0 \\
-1 &  \lambda \chi_0^3 P(\chi_0)^{-1} & 0 & 0 & 0 & -A_0 + \frac{\lambda}{2} \chi_0^2 \\
0 & - E_0 \lambda \chi_0^{-1} P(\chi_0)^{-1} & 0 & 0 & 0 & 0 \\
0 & - \lambda \chi_0  P(\chi_0)^{-1} & 0 & 0 & 0 & - \lambda \\
0 & 0 & - \lambda & 0 & -\lambda \chi_0^2 & 0
  \end{array} \right ) ,\\
    \mbox{Case 2}: \qquad &  F_{\xi_2} = \left ( \begin{array}{cccccc}
      0 & \frac{\lambda F_0}{8}  & \frac{\lambda E_0}{\chi} & 0 & \frac{1}{4} \lambda \chi \cos \theta_0 & 0 \\
-1 & 0 & 0 & - \frac{2 \lambda \chi^2}{F_0} \mbox{cotan} \theta_0 & 0 &
\frac{1}{2} \left ( \lambda \chi^2 - \frac{E_0^2}{\chi^4} \right ) \\
\frac{E_0}{\chi^3} & 0 & 0 & 0 & 0 & \frac{E_0 F_0}{8 \chi^3} \\
0 & 0 & 0 & 0 & - \lambda \chi \sin \theta_0 & 0 \\
0 & 0 & 0 & \frac{\lambda \chi }{\sin \theta_0} & 0 & 0 \\
0 & -\lambda & 0 & 0 & - \frac{2 \lambda \chi \cos \theta_0}{F_0} & 0 
    \end{array} \right ) ,\\
\mbox{Case 3}: \qquad &  F_{\xi_3} = \left ( \begin{array}{cccccc}
  0 & - \frac{\lambda}{2} \kappa & 0 & 0 & 0 & 0 \\
-1 & 0 & 0 & 0 & 0 & - \frac{\kappa}{2} \\
0 & 0 & 0 & 0 & 0 & 0 \\
0 & 0 & 0 & 0 & 0 & 0 \\
0 & 0 & 0 & 0 & 0 & 0 \\
0 & -\lambda & 0 & 0 & 0 & 0 
\end{array} \right ) ,
\end{align*}
where the point $p$ has cordinates $\{ \chi_0,u_0,x_0,y_0\}$ in case 1 and
coordinates $\{u_0, x_0, \theta_0, \phi_0 \}$ in case 2. (It is not necessary to give values to the coordinates of $p$ in case 3 as they do not explicitly appear in the matrix.)

In order to simplify the notation we denote $\mathcal{Q}_i := \mathcal{Q}_{F_{\xi_i}}$ the polynomials \eqref{defQF2} corresponding to the three cases. They are
given (up to an irrelevant global sign) by
  \begin{align*}
    \mbox{Case 1}: \qquad &  {\mathcal Q}_1 = z^3 + 2 A_0 \lambda z^2
                            - C_0 \lambda^3 z + E_0^2 \lambda^5,\\
                            \mbox{Case 2}: \qquad &  {\mathcal Q}_2  = \left ( z + \frac{\lambda E_0^2}{\chi^4} \right ) \left ( z - \lambda^2 \chi^2 \right )^2, \\
                            \mbox{Case 3}: \qquad &  {\mathcal Q}_3  =  z^2 \left ( z + \kappa \lambda \right ).                          
\end{align*}
Observe that these polynomials are independent of the point $p$. This is a necessary fact because the conformal class of the Killing vector is independent of
the point, and hence the algebraic classification of $F_{\xi}$ up to 
conjugation must also be independent of the point. This provides a non-trivial test both to the validity of Theorem \ref{endomorfism} and for the calculations in this section.

We know from Proposition \ref{defgammamu} (cf. \cite{marspeon21}) that
the classification of $F_{\xi_i}$ is given by the set of parameters $\mathrm{Roots} ( \mathcal{Q}_i) = \{-\mu_t^2,\mu_s^2,\mu^2\}$, which are the roots of $\mathcal{Q}_1$ sorted in a way determined by the causal character of $\ker (F_{\xi_i})$. Thus, it is necessary to establish case by case the connection between the parameters defining the metrics $\widetilde g_i$ and $\mathrm{Roots} ( \mathcal{Q}_i)$.


\bigskip
\underline{Case 1.}
\bigskip

First observe that
\begin{equation}\label{Q1Pchi}
 \mathcal{Q}_1(\widetilde \chi^2) = -\lambda^3 \widetilde \chi^2 P( \lambda^{-1} \widetilde \chi).
\end{equation}
Hence, the roots of $\mathcal{Q}_1$ are determined by the roots of the polynomial $\chi^2 P(\chi)$. The parameters $\{A_0,C_0,E_0 \}$ defining $P(\chi)$ are restricted by the condition of $P(\chi)$ being positive in some interval $I \in \mathbb{R}\backslash\{0\}$, which is clearly equivalent to imposing this same condition on the polynomial $\chi^2 P(\chi)$. Morever, $\chi^2 P(\chi)$ is an even polynomial which has negative dominant term, i.e.  $\chi^2 P(\chi) \rightarrow \infty$ as $\chi \rightarrow \pm \infty$. If $E_0 = 0$, then $\chi^2 P(\chi)$ has a root at $\chi = 0$ which is at least double. So if $\chi^2 P(\chi)$ is to be positive on $I$, it either has one positive  root (and its corresponding negative root) if $\{0\}$ is in the closure of $I$, or two positive roots (and their corresponding negative roots) otherwise (cf. Figure \ref{fig}). If $E_0 \neq 0$, then $\chi^2 P(\chi)$ is negative at the origin, so if $\chi^2 P(\chi)$ is to be positive on $I$, it must have exactly two positive different roots (and their corresponding negative roots) (cf. Figure \ref{fig}). 

All the above cases are covered by the polynomial decomposition
\begin{equation}\label{decompol}
 \chi^2 P(\chi)= -\lambda(\chi^2 - a^2)(\chi^2-b^2)(\chi^2+c^2)
\end{equation}
where $a^2,b^2,c^2$ are real constants satisfying
\begin{equation}\label{eqabc}
a^2 +b^2 - c^2 = 2\frac{A_0}{\lambda},\quad -a^2 b^2 + a^2 c^2 + b^2 c^2 = \frac{C_0}{\lambda},\quad a^2 b^2 c^2 = \frac{E_0^2}{\lambda},
\end{equation}
with $a^2\neq b^2$ (as otherwise the positivity of $\chi^2 P(\chi)$ cannot hold). 
After swapping $a$ and $b$ if necessary we may assume $a^2 > b^2 \geq 0$. The argument can be read in the opposite direction: only by choosing any constants $a,b,c$ such that $a \neq b$ and  $a^2 > b^2 \geq 0$ and defining the parameters $\{A_0,C_0,E_0 \}$ via \eqref{eqabc}, one obtains a function $P(\chi)$ which is positive in some interval $I \in \mathbb{R}\backslash\{0\}$.  

\begin{figure}[h!]   
 \begin{center}
   {\scriptsize(i)}\includegraphics[width=0.27\textwidth]{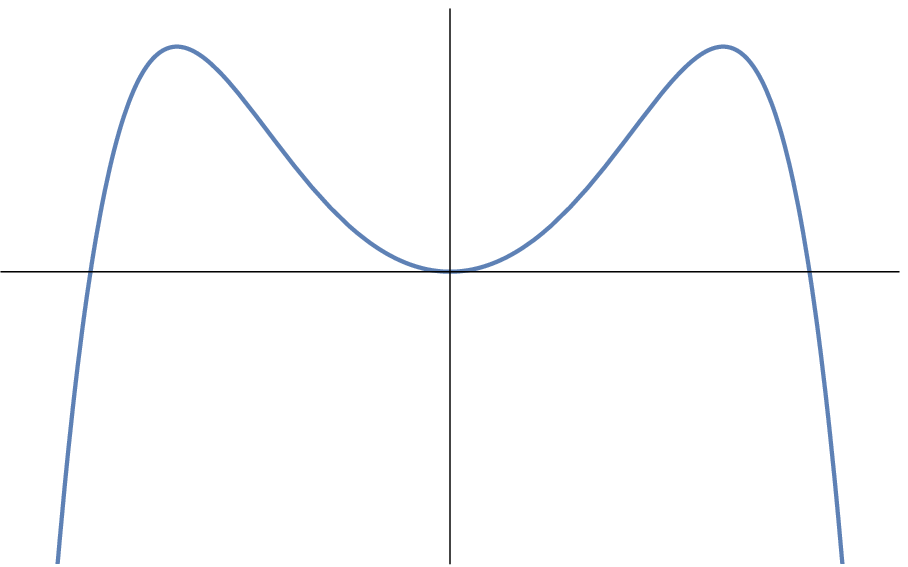}  
   \quad
  {\scriptsize(ii)}\includegraphics[width=0.27\textwidth]{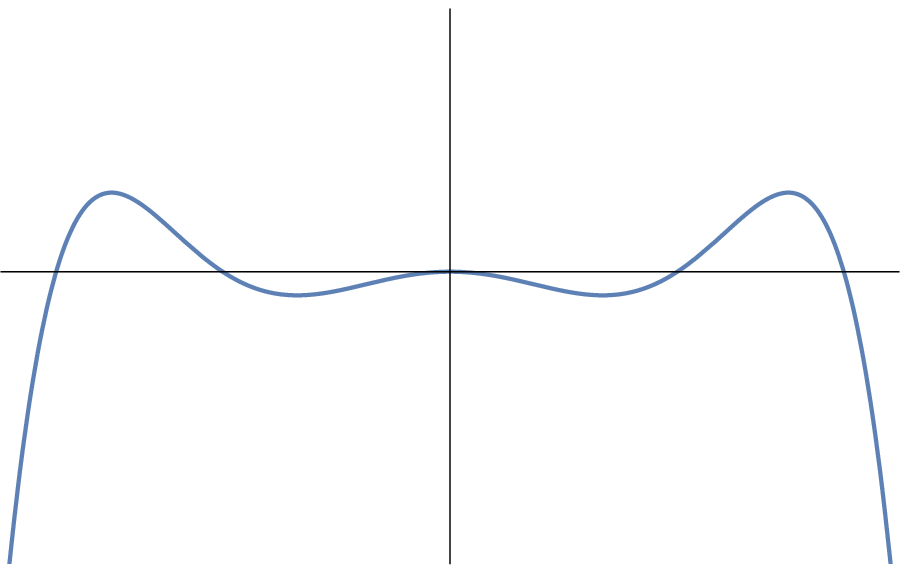}
   \quad
  {\scriptsize(iii)}\includegraphics[width=0.27\textwidth]{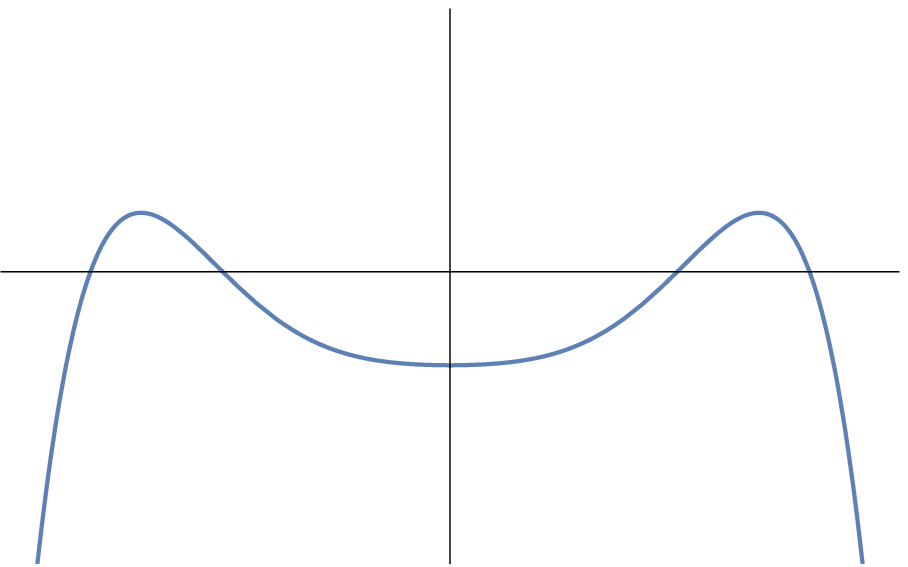}
 \caption{\small Possible profiles of the polynomial $\chi^ 2 P(\chi)$: $(i)$ $E_0 = 0$ with one real positive root, $(ii)$ $E_0 = 0$ with two real positive roots, $(iii)$ $E_0 \neq 0$.}\label{fig}     
 \end{center} 
 \end{figure}  

Therefore, by \eqref{Q1Pchi}, the roots of $\mathcal{Q}_1$ are always $\{ \lambda a^2, \lambda b^2, -\lambda c^2 \}$. According to Proposition \ref{defgammamu}, to determine their correspondence with $\{-\mu_t^2,\mu_s^2,\mu^2\}$ we must sort $\{ \lambda a^2, \lambda b^2, -\lambda c^2 \}$ depending on the causal character of $\ker F_{\xi_1}$. By Proposition \ref{defgammamu}, the sorting is such that  $-\mu_t^2 = \mu_s^2 =0$ and $\mu^2 = \lambda a^2$ whenever 
$\ker(F_{\xi_1})$ is null  (in which case, by Remark \ref{remarknullker}, $\mathcal Q_1$ has an at least double root at zero and no negative roots) and 
$-\mu_t^2 = -\lambda c^2, \mu_s^2 = \lambda a^2, \mu^2 = \lambda b^2$ when $\ker (F_{\xi_1})$ is non-null.


{If $E_0 \neq 0$, then zero is not a root of $\mathcal{Q}_1$. Being $\mathcal{Q}_1$ the square root of the charateristic polynomial of $F^2$ (see \eqref{defQF2}), this implies that $F_{\xi_1}$ has no zero eigenvalues.} Thus $\ker F_{\xi_1} = \{ 0 \}$ and by Proposition \ref{defgammamu}, the conformal class of $\xi_1$ is given by: 
\begin{itemize}
 \item[1.a)] $-\mu_t^2 = -\lambda c^2 $, $\mu_s^2 = \lambda a^2 $ and $\mu = \lambda b^2 $ covering the region $\{\mu_s^2 > \mu^2 > 0 > -\mu_t^2 \}$. 
\end{itemize}

When $E_0 = 0 $ it is straightforward to compute the kernel of $F_{\xi_1}$. The result
is
\begin{align*}
  \mbox{ker} (F_{\xi_1} ) = \mbox{span} \left \{ u^1 := v^4, u^2 := v^1 +
  v^2 + \frac{1}{\chi_0} v^3 - ( A_0 + \frac{\lambda}{2} \chi_0^2  ) v^6 \right \} ,
\end{align*}
so the kernel is two-dimensional in this case. The pull-back of the scalar product $Q$ on this space is
\begin{align*}
  Q(u^1,u^1)= \frac{1}{\chi_0^2}, \qquad
  Q(u^1, u^2 )=0, \qquad Q(u^2, u^2 ) = \frac{C_0}{\chi_0^2}.
  \end{align*}
This is timelike if $C_0 <0 $, null if $C_0=0$  and spacelike if $C_0>0$. By \eqref{eqabc}, $E_0  = 0$ requires the vanishing of $b^2$ and/or $c^2$,  and by Proposition \ref{defgammamu}, the conformal class of $\xi_1$ is given by
\begin{itemize}
 \item[1.b)] If $c^2 = 0, b^ 2 \neq 0$, then $C_0 < 0$ and $-\mu_t^2 = 0 $, $\mu_s^2 = \lambda a^2 $, $\mu = \lambda b^2 $ cover the region $\{\mu_s^2 > \mu^2 > 0 = -\mu_t^2 \}$.
  \item[1.c)] If $c^2 \neq 0, b^2 = 0$, then $C_0 > 0$ and $-\mu_t^2 = -\lambda c^2 $, $\mu_s^2 = \lambda a^2 $, $\mu = 0 $ cover the region $\{\mu_s^2 > \mu^2 = 0 > -\mu_t^2 \}$.
   \item[1.d)] If $c^2 = 0, b^ 2 = 0$, then $C_0 = 0$ and $-\mu_t^2 = 0 $, $\mu_s^2 = 0 $, $\mu = \lambda a^2 $ cover the region $\{ \mu^2 > 0 = \mu_s^2 = -\mu_t^2 \}$.
\end{itemize}


\bigskip
\underline{Case 2.}
\bigskip

In this case the roots of $\mathcal{Q}_2$ are immediately found to be $-{\lambda E_0^2}/{\chi^4}$ together with the double root $ \lambda^2 \chi^2$. If $E_0 \neq 0$  it again follows  that $\ker F_{\xi_2} = \{ 0 \}$, thus the conformal class of $\xi_2$ is given by the parameters $-\mu_t^2 = -{\lambda E_0^2}/{\chi^4} $, $\mu_s^2 = \mu = \lambda^2 \chi^2 $ covering the region $\{\mu_s^2 = \mu^2 > 0 > -\mu_t^2 \}$.  When $E_0 = 0$ the kernel of
$F_{\xi_2}$ is 
\begin{align*}
  \mbox{ker} (F_{\xi_2} ) = \mbox{span} \left \{ u^1 := v^3, u^2 := v^1 +
  \frac{\lambda \chi^2}{2} v^6 \right \}
\end{align*}
so the kernel is again two-dimensional and $Q$ restricted to 
this space is
\begin{align*}
  Q(u^1,u^1)= 1, \qquad
  Q(u^1, u^2 )=0, \qquad Q(u^2, u^2 ) = -\frac{4}{\lambda \chi^2} .
  \end{align*}
Thus, $\ker (F_{\xi_2} )$ is timelike and the conformal class of $\xi_2$ is given by 
\begin{itemize}
 \item[2)] $-\mu_t^2 = 0 $, $\mu_s^2 = \mu = \lambda^2 \chi^2 $ covering the region $\{\mu_s^2 = \mu^2 > 0 = -\mu_t^2 \}$.
\end{itemize}

\bigskip
\underline{Case 3.}
\bigskip

In this case the roots of $\mathcal{Q}_3$ are also trivial.  There is a double root at zero and another one at $-\lambda \kappa$.  Note that from the scaling freedom of $\xi$ (cf. \eqref{escfreedom}), $F_{\xi_3}$ is also defined up to a scaling factor $|\mu_0|^{1/2}$ and $F^2_{\xi_3}$ up to $|\mu_0|$. Then the root $-\lambda \kappa$ can be scaled by a non-zero positive factor $-\lambda |\mu_0| \kappa$, which will be relevant to cover the maximal space of parameters in the space of conformal classes. {Observe that (in any of the three cases) there is no restriction on the value of $\mu_0$.} 

The kernel of the endomorphism is
\begin{align*}
  \mbox{ker} (F_{\xi_3} ) = \mbox{span} \left \{ u^1 := v^1 - \frac{\kappa}{2} v^6, u^2 := v^3, u^3 := v^4, u^4:=v^5   \right \}
\end{align*}
hence four-dimensional and $Q$ restricted to 
this space is
\begin{align*}
  Q(u^1,u^1)= \kappa, \qquad
  Q(u^2, u^2 ) = 1, \qquad
Q(u^3, u^3 ) = \frac{1}{\Sigma^2|_p}, \qquad
Q(u^4, u^4 ) = \frac{1}{(\Sigma^2 \sin^2 \theta) |_p}
\end{align*}
and the rest zero, and 
where we have written the constant curvature metric as
\begin{align*}
  h = d \psi^2 + \Sigma^2 (\psi) \left ( d \theta^2 + \sin^2
  \theta d \phi^2 \right ).
\end{align*}
The kernel is now spacelike if $\kappa >0$, degenerate if $\kappa=0$
 and timelike if $\kappa<0$. Thus, for each case, the conformal class of $\xi_3$ is
\begin{itemize}
 \item[3.a)] For $\kappa = 1$, $-\mu_t^2 = -\lambda  |\mu_0|$, $\mu_s^2 = \mu = 0$ covering the region $\{\mu_s^2 = \mu^2 = 0 > -\mu_t^2 \}$.
  \item[3.b)] For $\kappa = 0$, $-\mu_t^2  = \mu_s^2 = \mu = 0$ covering a single point. 
   \item[3.c)] For $\kappa = -1$, $\mu_s^2 =\lambda  |\mu_0|$, $-\mu_t^2 =  \mu = 0$ covering the region $\{\mu_s^2 > \mu^2 = 0 = -\mu_t^2 \}$.
\end{itemize}

\bigskip
 \bigskip
 
 Summarizing, the cases 1,2 and 3 correspond to  CKVs $\xi_1, \xi_2$ and $\xi_3$ whose respective conformal classes cover the space of parameters 
 \begin{align}
  \mathcal{A}  =  & \{ (-\mu_t^2,\mu_s^2,\mu^2) \in \mathbb{R}^3 \mid \mu_s^2 \geq \mu^2 \geq 0 \geq -\mu_t^2, ~\mbox{with}~  \mu_t^2 ~\mbox{or} ~ \mu_s^2 \neq 0\} \\
  \bigcup & \{ (-\mu_t^2,\mu_s^2,\mu^2) \in \mathbb{R}^3 \mid  \mu^2 \geq 0 = \mu_s^2 = -\mu_t^2\}.
 \end{align}
By Proposition \eqref{defgammamu}, $\mathcal{A}$ corresponds to the entire space parameterizing the conformal classes of CKVs of conformally flat $4$-metrics. Thus, we have proven that these metrics correspond exactly to all metrics in the five dimensional Kerr-de Sitter-like class. This, combined with Theorem \ref{theoKS}, yields the following result:

\begin{theorem}
 In five spacetime dimensions, the following classes of $(\Lambda>0)$-vacuum metrics are equivalent:
 \begin{enumerate}
  \item  The Kerr-de Sitter-like class (cf. Definition \ref{defKdSlike}).
  \item The Kerr-Schild type metrics on a locally de Sitter background (cf. \eqref{KSchild}) admitting a smooth conformal extension such that $\Omega^2 \mathcal{H} \bm k \otimes \bm k = O(\Omega)$.
  \item The algebraically special metrics with non-degenerate optical matrix. 
 \end{enumerate}

\end{theorem}

\section{Discussion} \label{secdiscus}

We have obtained a method to determine the conformal class of an arbitrary CKV $\xi$ of a locally conformally flat metric $\gamma$ of any dimension and signature. Such method is based on pointwise properties of the CKVs and it is independent on the coordinates and the representative of the class of metrics conformal to $\gamma$. This improves previously existing results (cf. \cite{marspeondata21,marspeonKSKdS21}) which require to find an explicitly flat representative in Cartesian coordinates. 

Our result is stated as a computationally neat algorithm in Theorem \ref{theoendocords}, which allows for a straightforward application in Section \ref{secspecial}. Namely, we classify the asymptotic data of all five-dimensional, algebraically special, $(\Lambda>0)$-vacuum spacetimes, whose optical matrix is non-degenerate (cf. \cite{reall15}). Such asymptotic data are determined by the conformal class of a CKV of a conformally flat $\scri$. Furthermore, we prove equivalence of this collection of spacetimes with the Kerr-de Sitter-like class as well as with the $(\Lambda>0)$-vacuum  Kerr-Schild spacetimes satisfying a natural asymptotic condition (cf. Theorem \ref{theoKS}).

It is worth commenting that the results of Section \ref{secspecial}, besides providing an application of Theorem \ref{theoendocords}, outline the way for potential future results.  As pointed out in \cite{reall15}, in dimensions higher than four, there exist several results (cf. \cite{diasreall13,ort12}) supporting the idea that the class of algebraically special solutions is more rigid than in four dimensions.  One then wonders whether Theorem \ref{theoKS} also holds in any dimension higher than five. 
Surprisingly, the results in \cite{reall15} do not rely on any asymptotic property of the spacetime, while in \cite{marspeonKSKdS21} it is central. Yet, the class of spacetimes studied \cite{reall15} and in \cite{marspeonKSKdS21} (in five spacetime dimensions) happen to be equivalent, as we have shown in this paper. This hints a possible connection between asymptotic properties of spacetimes and the algebraic classification of the Weyl tensor. More precisely, the algebraically special condition with non-degenerate optical matrix implies conformal extendability with locally conformally flat $\scri$. 
A better understanding of this aspect would be of substantial intrinsic interest and key for extending Theorem \ref{theoKS} to arbitrary dimensions.


It is also interesting to observe that the Weyl tensor $C$ of the spacetime determines, in any dimension, the Weyl tensor\footnote{\car{Note that $c$ and $D$ as defined in \eqref{rescweyl}, are indeed independent objects at $\scri$, however both obtainable from the spacetime Weyl tensor $C$ at $\scri$.}} $c$ of the  metric $\gamma$ induced at $\scri$. Indeed, an asymptotic expansion of $C$, shows that its components fully tangent to $\scri$ coincide with $c$ to the leading order. The Weyl tensor contains a lot of information about the conformal class of metrics of dimension equal or higher than four\footnote{The conformal class determines the Weyl tensor, but the opposite is not always true, cf. \cite{hall09}.}, so it would not be surprising that the algebraically special condition on $C$ imposes strong conditions on the conformal class of (four or higher dimensional) $\gamma$. However, in four spacetime dimensions, $\gamma$ is three dimensional and thus $c$ vanishes identically, so it is not clear in this case how the algebraic type of $C$ affects the conformal class of $\gamma$. The difference between the four and higher dimensional cases may be responsible for the lack of rigidity of algebraically special metrics in four dimensions, because recall (cf. Section \ref{secKdSlike}) that the conformal class of $\gamma$ is one of the freely speciable data in the asymptotic Cauchy problem. This, however, does not rule out other possible relations between the algebraic type of the four-dimensional spacetime metrics and their asymptotic properties. One connection may arise from the fact that the other component of the asymptotic data is, in four dimensions, the electric part of the rescaled Weyl tensor (cf.  Section \ref{secKdSlike}). In addition, constraints on the conformal class of $\gamma$ may also appear as a consequence of the relation between $C$ and Cotton tensor of $\gamma$, which plays a similar role than the Weyl tensor for three dimensional metrics. 
 These potential connections are worth to investigate in the future.

  \section{Acknowledgments}
  The authors wish to thank Igor Khavkine for comments on the manuscript and for providing
useful references. MM acknowledges financial support under the projects PGC2018-096038-B-I00 (Spanish Ministerio de Ciencia, Innovaci\'on y Universidades and FEDER), SA096P20 (Junta de Castilla y Le\'on) and CPN is supported by the grant No 22-14791S of the Czech Science Foundation.


\begin{thebibliography}{10}
  
\bibitem{Anderson2005}
M.~T. Anderson.
\newblock Existence and stability of even-dimensional asymptotically de
  \mbox{Sitter} spaces.
\newblock {\em Annales Henri Poincar{\'e}}, 6:801--820, 2005.

\bibitem{andersonchrusciel05}
M.~T. Anderson and P.~T. Chruściel.
\newblock Asymptotically simple solutions of the vacuum \mbox{Einstein}
  equations in even dimensions.
\newblock {\em Communications in Mathematical Physics}, 260:557–577, 2005.

\bibitem{BeigChruscielKID97}
R.~Beig and P.~T. Chru\'sciel.
\newblock Killing initial data.
\newblock {\em Classical and Quantum Gravity}, 14:83–92, 1997.


\bibitem{reall15}
G.~Bernardi~de Freitas, M.~Godazgar, and H.~S. Reall.
\newblock {Uniqueness of the \mbox{Kerr-de Sitter} spacetime as an
  algebraically special solution in five dimensions}.
\newblock {\em Communications in Mathematical Physics}, 340:291--323, 2015.

\bibitem{Freitas16}
G.~Bernardi~de Freitas, M.~Godazgar, and H.~S. Reall.
\newblock Twisting algebraically special solutions in five dimensions.
\newblock {\em Classical and Quantum Gravity}, 33:095002, 2016.

\bibitem{blairconformal}
D.~E. Blair.
\newblock {\em Inversion Theory and Conformal Mapping}.
\newblock Student Mathematical Library. American Mathematical Society,
  {Providence, Rhode Island}, 2000.

  \bibitem{Chenteo11}
Y. Chen, and E. Teo.
\newblock A new AF gravitational instanton.
\newblock {\em Physics Letters B}, 703: 359-362, 2011.

  
\bibitem{coley04}
A.~Coley, R.~Milson, V.~Pravda, and A.~Pravdová.
\newblock Classification of the \mbox{Weyl} tensor in higher dimensions.
\newblock {\em Classical and Quantum Gravity}, 21:35–41, 2004.

\bibitem{adjointorbs}
D. H. Collingwood and W. L. McGovern.
\newblock {\em Nilpotent orbits in semisimple Lie algebras}.
\newblock Van Nostrand Reinhold, New York 1993.

\bibitem{Der12}
A.~Derdzinski.
\newblock Two-jets of conformal fields along their zero sets.
\newblock {\em Central European Journal of Mathematics}, 10:1698--1709, 2012.

\bibitem{diasreall13}
O.~J.~C. Dias and H.~S. Reall.
\newblock Algebraically special perturbations of the Schwarzschild solution in
  higher dimensions.
\newblock {\em Classical and Quantum Gravity}, 30:095003, 2013.

\bibitem{FeffGrah85}
C.~Fefferman and C.~R. Graham.
\newblock Conformal invariants.
\newblock {\em \'Elie Cartan et les math\'ematiques d'aujourd'hui}, 
  S131: 95--116. Soci\'et\'e math\'ematique de France, 1985.

\bibitem{ambientmetric}
C.~Fefferman and C.~R. Graham.
\newblock {\em The Ambient Metric}, {\em Annals of Mathematics
  Studies}, 178.
\newblock Princeton University Press, 2012.

\bibitem{friedrich81bis}
H.~Friedrich.
\newblock {On the regular and asymptotic characteristic initial value problem
  for Einstein's vacuum field equations}.
\newblock {\em Proceedings of the Royal Society of London A}, 375:169--184,
  1981.

\bibitem{friedrich81}
H.~Friedrich.
\newblock {The asymptotic characteristic initial value problem for Einstein's
  vacuum field equations as an initial value problem for a first-order
  quasilinear symmetric hyperbolic system}.
\newblock {\em Proceedings of the Royal Society of London A}, 378:401--421,
  1981.

\bibitem{Fried86geodcomp}
H.~Friedrich.
\newblock On the existence of n-geodesically complete or future complete
  solutions of einstein's field equations with smooth asymptotic structure.
\newblock {\em Communications in Mathematical Physics}, 107:587--609, 1986.

\bibitem{Gibbons2005}
G.W. Gibbons, H.~Lü, D.~N. Page, and C.N. Pope.
\newblock The general \mbox{Kerr–de Sitter} metrics in all dimensions.
\newblock {\em Journal of Geometry and Physics}, 53:49 -- 73, 2005.

\bibitem{hall09}
G.~Hall.
\newblock Some remarks on the converse of {Weyl's} conformal theorem.
\newblock {\em Journal of Geometry and Physics}, 60:1--7, 2010.

\bibitem{Holl05}
S.~Hollands, A.~Ishibashi, and D.~Marolf.
\newblock Comparison between various notions of conserved charges in
  asymptotically \mbox{AdS} spacetimes.
\newblock {\em Classical and Quantum Gravity}, 22:2881–2920, 2005.

\bibitem{Kaminski21}
W.~Kamiński.
\newblock Well-posedness of the ambient metric equations and stability of even
  dimensional asymptotically de \mbox{Sitter} spacetimes, 2021.  	ArXiv:2108.08085.

\bibitem{kichenassamy03}
S.~Kichenassamy.
\newblock On a conjecture of \mbox{Fefferman and Graham}.
\newblock {\em Advances in Mathematics}, 184:268--288, 2004.

  
\bibitem{knapp}
A.W. Knapp.
\newblock {\em Lie Groups Beyond an Introduction}.
\newblock Progress in Mathematics. Birkhäuser Boston, 2002.

\bibitem{mars99}
M.~Mars.
\newblock A spacetime characterization of the \mbox{Kerr} metric.
\newblock {\em Classical and Quantum Gravity}, 16:2507–2523, 1999.

\bibitem{Kdslike}
M.~Mars, T.~T. Paetz, and J.~M.~M. Senovilla.
\newblock {Classification of Kerr–de Sitter-like spacetimes with conformally
  flat {$\mathscr{I}$}}.
\newblock {\em Classical and Quantum Gravity}, 34:095010, 2017.

\bibitem{KdSnullinfty}
M.~Mars, T.T. Paetz, J.~M.~M. Senovilla, and W.~Simon.
\newblock Characterization of (asymptotically) \mbox{Kerr–de Sitter-like}
  spacetimes at null infinity.
\newblock {\em Classical and Quantum Gravity}, 33:155001, 2016.

\bibitem{marspeon20}
M.~Mars and C.~Pe\'on-Nieto.
\newblock Skew-symmetric endomorphisms in $\mathbb{M}^{1,3}$: a unified
  canonical form with applications to conformal geometry.
\newblock {\em Classical and Quantum Gravity}, 38:035005, 2020.

\bibitem{marspeondata21}
M.~Mars and C.~Peón-Nieto.
\newblock Free data at spacelike $\mathscr{I}$ and characterization of {Kerr-de
  Sitter} in all dimensions.
\newblock {\em The European Physical Journal C}, 81, 2021.

\bibitem{marspeon21}
M.~Mars and C.~Peón-Nieto.
\newblock Skew-symmetric endomorphisms in $\mathbb{M}^{1,n}$: a unified
  canonical form with applications to conformal geometry.
\newblock {\em Classical and Quantum Gravity}, 38:125009, 2021.

\bibitem{marspeonKSKdS21}
M.~Mars and C.~Peón-Nieto.
\newblock Classification of {Kerr-de Sitter-like} spacetimes with conformally
  flat $\mathscr{I}$ in all dimensions.
\newblock {\em Physical Review D}, 105:044027, 2022.

\bibitem{marsseno15}
M.~Mars and J.~M.~M. Senovilla.
\newblock {A spacetime characterization of the Kerr-NUT-(A)de Sitter and
  related metrics}.
\newblock {\em Annales Henri Poincaré}, 16:1509--1550, 2015.

\bibitem{ort12}
M.~Ortaggio, V.~Pravda, A.~Pravdov{\'{a}}, and H.~S. Reall.
\newblock On a five-dimensional version of the {Goldberg{\textendash}Sachs}
  theorem.
\newblock {\em Classical and Quantum Gravity}, 29:205002, 2012.

\bibitem{Ortaggio12}
M.~Ortaggio, V.~Pravda, and A.~Pravdová.
\newblock Algebraic classification of higher dimensional spacetimes based on
  null alignment.
\newblock {\em Classical and Quantum Gravity}, 30:013001, 2012.

\bibitem{KIDPaetz}
T.~T. Paetz.
\newblock {Killing Initial Data} on spacelike conformal boundaries.
\newblock {\em Journal of Geometry and Physics}, 106:51 -- 69, 2016.

\bibitem{IntroCFTschBook}
M.~Schottenloher.
\newblock {\em A Mathematical Introduction to Conformal Field Theory},
    {\em Lecture Notes in Physics}, 43.
\newblock Springer, Berlin-Heidelberg, 2008.

\bibitem{simon84}
W.~{Simon}.
\newblock {Characterizations of the Kerr metric.}
\newblock {\em General Relativity and Gravitation}, 16:465--476, 1984.

\bibitem{kroonbook}
J.~A. Valiente-Kroon.
\newblock {\em {Conformal methods in general relativity}}.
\newblock Cambridge University Press, Cambridge, 2016.

\end{thebibliography}

\end{document}